\documentclass[12pt,a4paper]{article}
\usepackage{geometry}
\usepackage{amsmath, amsthm, amssymb, bm, graphicx, mathrsfs, multirow}
\usepackage{xcolor}
\usepackage[authoryear]{natbib}
\usepackage{indentfirst}
\usepackage[ruled,vlined]{algorithm2e}
\DontPrintSemicolon
\SetKwComment{tcp}{ }{}
\usepackage[colorlinks = true, allcolors = blue]{hyperref}
\usepackage{caption}
\usepackage{booktabs}
\usepackage{siunitx}
\usepackage{enumitem}
\usepackage{dcolumn}
\usepackage{dsfont}
\usepackage{array}
\usepackage{setspace}
\usepackage{chngcntr} 

\newcommand{\beginsupplement}{%
\counterwithout{figure}{section}
\counterwithout{table}{section}

\setcounter{figure}{0}
\setcounter{table}{0}

\renewcommand{\thefigure}{S.\arabic{figure}}
\renewcommand{\thetable}{S.\arabic{table}}

\providecommand{\theHfigure}{}
\providecommand{\theHtable}{}
\renewcommand{\theHfigure}{S.\arabic{figure}}
\renewcommand{\theHtable}{S.\arabic{table}}
}

\definecolor{color1}{HTML}{a89ee0}
\definecolor{color2}{HTML}{fff290}
\definecolor{color3}{HTML}{a1ff90}
\newcolumntype{d}[1]{D{.}{.}{#1}}
\newcolumntype{C}[1]{>{\centering}m{#1}}

\def\bfbeta{\bm{\beta}}

\def\rmT{\top}
\DeclareMathOperator*{\argmax}{arg\,max}

\newcommand{\LRs}[1]{\left(#1\right)}
\newcommand{\LRm}[1]{\left[#1\right]}
\newcommand{\LRl}[1]{\left\{#1\right\}}

\def\E{\mathbb{E}}
\def\indep{\perp\!\!\!\perp}

\theoremstyle{plain}

\newtheorem{thm}{Theorem}
\newtheorem{prop}{Proposition}
\newtheorem{lemma}{Lemma}

\theoremstyle{remark}
\newtheorem{remark}{Remark}
\newtheoremstyle{examplestyle}  
{3pt}                          
{3pt}                          
{\normalfont}                  
{\parindent}                   
{\itshape}                     
{.}                            
{.5em}                         
{\thmname{#1}\thmnumber{ #2}\thmnote{ \textit{(#3)}}}  

\theoremstyle{examplestyle}
\newtheorem{example}{Example}

\title{\textbf{Model-Agnostic and Uncertainty-Aware Dimensionality Reduction in Supervised Learning}}
\author{Yue Yu$^{1}$,\ Guanghui Wang$^{1}$,\ Liu Liu$^{2}$,\ and Changliang Zou$^{1}$             \\
{$^{1}${\small\it School of Statistics and Data Science, Nankai University, Tianjin, China}} \\
{$^{2}${\small\it School of Mathematical Science, Chengdu University of Technology, Chengdu, China}}
}
\date{}

\def\V{\mathcal{L}}

\def\bfX{X}
\def\bfY{Y}
\def\bfB{B}
\def\bff{f}
\def\bfv{v}
\def\bfH{H}
\def\bfu{u}
\def\bfA{A}
\def\bfW{W}
\def\bfR{R}
\def\bfx{x}
\def\bfZ{Z}
\def\P{\mathbb{P}}
\def\Q{\mathbb{Q}}

\begin{document}
\maketitle

\begin{abstract}
Dimension reduction is a fundamental tool for analyzing high-dimensional data in supervised learning. Traditional methods for estimating intrinsic order often prioritize model-specific structural assumptions over predictive utility. This paper introduces predictive order determination (POD), a model-agnostic framework that determines the minimal predictively sufficient dimension by directly evaluating out-of-sample predictiveness. POD quantifies uncertainty via error bounds for over- and underestimation and achieves consistency under mild conditions. By unifying dimension reduction with predictive performance, POD applies flexibly across diverse reduction tasks and supervised learners. Simulations and real-data analyses show that POD delivers accurate, uncertainty-aware order estimates, making it a versatile component for prediction-centric pipelines.
\end{abstract}

\section{Introduction}\label{sec:intro}

Dimension reduction is a core idea in modern statistics. It provides low-dimensional summaries of high-dimensional data while keeping the information needed for reliable prediction and inference \citep{jolliffe2002principal,li2018sufficient}. In supervised problems, we observe predictors $\bfX\in\mathbb{R}^p$ and aim to predict a response $\bfY\in\mathbb{R}^q$ using as few transformed variables as possible. The smallest such number is the \textit{intrinsic order}, denoted $d^*$.

\subsection{Motivating examples}

\begin{example}[Factor regression]\label{eg:factor regression}
Latent factors $\bff\in\mathbb{R}^{d^*_{\rm F}}$ drive both the predictors $\bfX$ and, through an unknown link $m$, the response $\bfY$ with $q=1$ \citep{fan2023factor,fan2024latent}:
\[
Y=m(\bff)+\varepsilon,\qquad
\bfX = \bfB \bff + \bfu,
\]
where $\bfB$ is a loading matrix and $\varepsilon$ and $\bfu$ are noise. All predictive information resides in the $d^*_{\rm F}$ coordinates of $\bff$. When $m$ is linear, the model reduces to principal component regression. 
\end{example}

\begin{example}[Sufficient dimension reduction, SDR]\label{eg:SDR}
There may exist a matrix $\bfB\in\mathbb{R}^{p\times d^*_{\rm CS}}$ with the smallest possible number of columns such that $\bfY$ (with $q=1$) is independent of $\bfX$ given $\bfB^{\top}\bfX$. The span of $\bfB$ is the central subspace and has dimension $d^*_{\rm CS}$ \citep{li1991sliced,cook1991sliced}. Focusing only on $\E(Y\mid \bfX)$ yields the central mean subspace, with dimension $d^*_{\rm CMS}\le d^*_{\rm CS}$ \citep{cook2002dimension}.
\end{example}

\begin{example}[Reduced-rank regression]\label{eg:reduced rank regression}
For multivariate responses,
\[
\bfY=\bfA^{\top}\bfX+\varepsilon,\qquad
\bfA=\bfB\Omega,
\]
where $\bfA$ has rank $d^*_{\rm R}<\min\{p,q\}$, $\bfB\in\mathbb{R}^{p\times d^*_{\rm R}}$, $\Omega\in\mathbb{R}^{d^*_{\rm R}\times q}$, and $\varepsilon$ is noise. All predictive information lies in the $d^*_{\rm R}$ coordinates of $\bfB^{\top}\bfX$ \citep{izenman1975reduced}.
\end{example}

Though these paradigms differ mechanistically, they share a unified goal: finding the minimal dimension \(d^*\) (e.g., $d^*_{\rm F},d^*_{\rm CS},d^*_{\rm CMS},d^*_{\rm R}$) that preserves all predictive information in \(X\) for \(Y\). This common thread motivates a shift in perspective: rather than relying on problem-specific structural diagnostics, we seek to determine \(d^*\) by directly measuring how dimensionality affects predictive performance.

\subsection{Limitations of classical order-estimation methods}

Classical approaches generally frame order estimation as a model selection problem. Information criteria based on penalized likelihood or eigenvalue thresholds remain popular for factor models \citep{bai2002determining}, SDR \citep{zhu2006sliced}, and reduced-rank regression \citep{bunea2011optimal}; adaptive penalties have also been proposed for SDR \citep{luo2016combining}. Other common approaches include cross-validation, used in SDR \citep{xia2002adaptive} and reduced-rank regression \citep{yuan2007dimension}, as well as sequential tests comparing eigenvalues from method-specific kernel matrices in SDR \citep[e.g.,][]{li1991sliced,bura2011dimension} or from the sample covariance matrix in factor models (typically in an \textit{unsupervised} fashion) \citep[e.g.,][]{onatski2009testing,yu2025testing}. Visual heuristics such as scree plots or regression plots persist in practice despite limited theoretical support \citep{cattell1966scree,cook1994interpretation}.

These approaches typically prioritize structure assumptions (e.g., eigenvalue decay or distributional conditions) over the ultimate goal of supervised learning: prediction accuracy. While they yield point estimators under specific combinations of task, model, and algorithm, they lack a unified framework to evaluate dimensionality through the lens of predictive sufficiency. Consequently, they provide limited guidance for downstream prediction tasks and often do not quantify \textit{uncertainty} in a systematic way. Recent studies highlight the practical importance of uncertainty quantification in order estimation. For instance, in factor regression, mild overestimation of the number of factors may be relatively benign \citep{fan2022learning}; however, underestimation risks missing important predictive signals, and severe overestimation accumulates noise, both of which harm predictive performance. Moreover, as modern data-analysis pipelines increasingly incorporate complex, non-linear dimension-reduction methods such as deep neural networks \citep{fan2023factor,chen2024deep}, traditional eigenvalue-based diagnostics become unreliable due to the resulting complex spectral behavior. These challenges demand a framework that explicitly ties dimension reduction to predictive performance, provides model-agnostic uncertainty quantification, and prioritizes supervision-aware dimension selection.

\subsection{Our contributions}\label{sec:contribution}

To address these challenges, we introduce \textit{predictive order determination (POD)} and contribute two main advances.
\begin{itemize}
\item \textbf{Prediction-optimal dimension via supervision.} For any candidate dimension \(d\), POD evaluates the predictiveness gap, that is, the drop in out-of-sample prediction risk incurred by using only \(d\) dimensions instead of the full representation. Under mild conditions, the smallest \(d\) at which this gap vanishes is \(d^*\). Predictiveness thus quantifies the utility of additional coordinates for the supervised task and decouples order determination from specific algorithms or distributional assumptions.

\item \textbf{Error bounds and consistency guarantees.} POD sequentially tests whether a given dimension is sufficient via cross-fitted predictiveness comparisons. This procedure (i) controls the probability of overestimating $d^*$ at a pre-specified level $\alpha$, (ii) bounds the probability of underestimation, and (iii) achieves consistency under regularity conditions.
\end{itemize}

\subsection{Related literature}

Prediction-driven order selection has been explored, but mostly within model-specific frameworks. In SDR with a multi-index structure, \cite{xia2002adaptive} proposed minimum average variance estimation, which estimates the central mean subspace by iterating local linear smoothing with index updates; the structural dimension is chosen by cross-validation, yielding a consistent point estimator. In multivariate linear regression, \cite{yuan2007dimension} obtained reduced-rank estimators via penalized least squares and used generalized cross-validation as a practical rule to select the penalty and hence the rank. Despite their predictive focus, these procedures remain tied to particular models and estimators and offer limited uncertainty quantification for the selected order. In contrast, our approach is model- and algorithm-agnostic and provides error rate control beyond consistency.

\cite{lei2020cross} introduced cross-validation with confidence, which conditions on a fixed set of candidate models and returns a confidence set of near-optimal models. When candidate risks are equal, however, its test statistic can \textit{degenerate} \citep{chen2024zipper}. POD instead targets population risk indexed by dimension, remains agnostic to the model and learning algorithm, and guarantees error rate control even under such degeneracy.

A parallel, algorithm-agnostic literature on \textit{variable importance} formalizes a predictiveness contrast: compare the population risk using all features with the risk when a subset is omitted, and develop nonparametric, efficient estimators and valid tests for that target \citep{williamson2023general,dai2022significance}. Our focus differs in both target and setting. Rather than asking ``which features matter,'' we ask ``how many dimensions suffice''. POD infers the minimal predictive order by detecting where predictiveness levels off as dimension grows. Moreover, unlike variable-importance frameworks that operate directly on the observed predictors $\bfX$, POD assesses predictiveness through an \textit{unobserved} low-dimensional representation that is learned as a surrogate. Our theory therefore quantifies how representation error transfers through the cross-fitted predictiveness gaps and the sequential test, leading to uncertainty control (bounds on overestimation and underestimation) and consistency under mild regularity.

\subsection{Structure}

The rest of the paper is structured as follows. Section \ref{sec:method} formalizes the concept of predictiveness-induced order, develops our cross-fitted tests, and presents the sequential procedure for estimating $d^*$. Section \ref{sec:theory} establishes error bounds and consistency results. Section \ref{sec:num} provides simulation studies and a real-data illustration. Section \ref{sec:conclusion} concludes this paper. Technical proofs and additional examples appear in the Supplementary Material.

\textbf{Notation.} For a positive integer $n$, write $[n]=\{1,\dots,n\}$. For a vector $x\in\mathbb{R}^p$ and an index set $l\subseteq[p]$, let $x_l$ collect the coordinates in $l$ and $x_{-l}$ the remainder. The Euclidean norm is $\|\bfx\|_2=(\sum_{i=1}^p x_i^2)^{1/2}$. For an event $A$, $1\{A\}$ denotes its indicator. For $a\in\mathbb{R}$, let $\lfloor a\rfloor$ denote the greatest integer less than or equal to $a$. The symbol $\delta_z$ is the point-mass distribution at $z$.

\section{Methodology}\label{sec:method}

\subsection{Predictiveness-induced order}\label{sec:order}

\subsubsection{Oracle and surrogate representations}

Let $(\bfX,\bfY)\sim\Q$ on $\mathcal{X}\times\mathcal{Y}$. We assume there exists an unobserved \textit{oracle representation} $$\bfR^*\in\mathbb{R}^{d^*},$$ with unknown intrinsic order $d^*$, such that $\bfR^*$ contains all information in $\bfX$ that is relevant for predicting $\bfY$. This setup covers, for example, factor regression with $\bfR^*=\bff$ (Example~\ref{eg:factor regression}), and SDR or reduced-rank regression with $\bfR^*=\bfB^{\top}\bfX$ (Examples~\ref{eg:SDR} and \ref{eg:reduced rank regression}). To compare candidate orders within a common ambient space, we fix a user-chosen upper bound $d_{\max}\geq d^*$ and embed $\bfR^*$ into $\mathbb{R}^{d_{\max}}$ by zero-padding:$$\bfR\in\mathbb{R}^{d_{\max}}\text{ with } \bfR_{[d^*]}=\bfR^*\text{ and }\bfR_{-[d^*]}=0.$$ Let $\bfZ=(\bfR,\bfX,\bfY)$ and let $\P$ denote the joint distribution of $\bfZ$. The distribution $\P$ is used only to define population quantities involving $\bfR$; the observed data are independent and identically distributed (i.i.d.) draws from $\Q$.

In practice, $\bfR$ is unobserved. We instead construct a data-dependent \textit{surrogate representation} $$\widehat{\bfR}=\widehat{\varphi}(\bfX)\in\mathbb{R}^{d_{\max}},$$ where $\widehat{\varphi}:\mathcal{X}\rightarrow\mathbb{R}^{d_{\max}}$ is produced by a chosen dimension-reduction method; for example, by projecting $\bfX$ onto leading eigenvectors of a sample covariance or kernel matrix \citep{xia2002adaptive,fan2023factor}. We assume the coordinates of $\widehat{\bfR}$ are ordered by an importance score returned by the reduction method, such as decreasing sample eigenvalues \citep{bai2002determining}, so that earlier coordinates are intended to be more informative.

\subsubsection{Predictiveness}

Let $\mathcal{G}=\{g:\mathbb{R}^{d_{\max}}\to \widehat{\mathcal{Y}}\}$ be a class of prediction rules (typically $\widehat{\mathcal{Y}}=\mathcal{Y}$). For $Y\in\mathcal{Y}$ and $\widehat{Y}\in\widehat{\mathcal{Y}}$, let $\ell(\bfY,\widehat{Y})$ be a loss function. Examples include squared loss $\ell(\bfY,\widehat{Y})=\|\bfY-\widehat{Y}\|_2^2$ for continuous responses with $\mathcal{Y}=\widehat{\mathcal{Y}}=\mathbb{R}^q$, 0-1 loss $\ell(\bfY,\widehat{Y})=1\{Y\neq\widehat{Y}\}$ for binary classification with $\mathcal{Y}=\widehat{\mathcal{Y}}=\{0,1\}$, and cross-entropy $\ell(\bfY,\widehat{Y})=-e_{Y}^{\top}\log\widehat{Y}$ for multiclass outcomes when $\mathcal{Y}=[M]$ and $\widehat{\mathcal{Y}}=[0,1]^M$, where $e_m$ the $m$th standard basis vector. For $g\in\mathcal{G}$, define the population risk $$\V(g,\P)= \E_{\P}\{\ell(\bfY,g(\bfR))\},$$ with smaller values indicating better predictive performance.

For $0\le d\le d_{\max}$, define the subclass of rules that depend on $\bfR$ only through its first $d$ coordinates: $$\mathcal{G}_d=\left\{g\in\mathcal{G}:g(r)=g(r')\text{ whenever }r_{[d]}=r'_{[d]}\right\}.$$ We define the \textit{predictiveness at order $d$} by
\[
\V_d=\min_{g\in\mathcal{G}_d}\V(g,\P),
\]
and let
\begin{equation*}
g_d\in{\arg\min}_{g\in\mathcal{G}_d}\V(g,\P)
\end{equation*}
denote a minimizer.

\begin{prop}\label{prop: equivalence of V}
$\V_0\geq \V_1\geq \dots \geq \V_{d^*}=\V_{d^*+1}=\dots=\V_{d_{\max}}$.
\end{prop}

Proposition \ref{prop: equivalence of V} shows that $\{\V_d\}_{d=0}^{d_{\max}}$ is non-increasing and becomes constant once all informative coordinates are included. We therefore define the target order
\[
d^*(\V)=\min\{0\leq d\leq d_{\max}:\V_d=\V_{d_{\max}}\}.
\]
Such a definition is inherently meaningful in its own right whenever the ultimate goal is prediction-oriented, as it directly quantifies the minimal complexity needed for optimal supervised performance. Under mild regularity conditions and suitable losses, $d^*(\V)$ coincides with intrinsic orders in the motivating examples. For continuous responses with squared loss, $d^*(\V)=d^*_{\rm F}$ in Example \ref{eg:factor regression}, $d^*(\V)=d^*_{\rm CMS}$ in Example \ref{eg:SDR}, and $d^*(\V)=d^*_{\rm R}$ in Example \ref{eg:reduced rank regression}, whereas with negative log-likelihood, $d^*(\V)= d^*_{\rm CS}$ in Example \ref{eg:SDR}. For categorical responses, cross-entropy also implies $d^*(\V)=d^*_{\rm CS}$ in Example \ref{eg:SDR}. See Section \ref{subsec:d*_d0} of the Supplementary Material for proofs and additional examples. Henceforth, we write $d^*=d^*(\V)$ and omit the explicit dependence on $\V$.

\begin{remark}
In Example \ref{eg:SDR}, squared loss targets the conditional mean and can yield $d^*(\V)=d^*_{\rm CMS}<d^*_{\rm CS}$ when the central mean subspace is strictly contained in the central space. This mirrors the distinction between conditional mean independence and conditional independence \citep{lundborg2024projected}. When the goal is prediction under the chosen loss, $d^*(\V)$ remains the appropriate target.
\end{remark}

\subsection{Testing whether a candidate order is sufficient}\label{sec:estimation-inference}

For each $d\in\{0,\ldots,d_{\max}\}$, we test whether the first $d$ coordinates already achieve maximal predictiveness:
\[
H_{0,d}:\V_d=\V_{d_{\max}}\quad\text{versus}\quad H_{1,d}:\V_d>\V_{d_{\max}}.
\]
Equivalently, we test whether the predictiveness gap $\V_d-\V_{d_{\max}}$ equals zero.

\subsubsection{Cross-fitting procedure}\label{sec:two-step}

Let $\{(\bfX_i,\bfY_i):i\in[n]\}$ be i.i.d.\ observations from $\Q$. The oracle representation $\bfR$ and the optimal prediction rules $g_d$ and $g_{d_{\max}}$ are unobserved. We therefore learn a reduction map, compute surrogate representations, and fit rules on the surrogate coordinates. If we were to evaluate predictive performance on the same data used for training, the resulting estimates would generally be biased and their asymptotic behavior difficult to analyze. To avoid this ``double dipping'', we use $K$-fold \textit{cross-fitting} \citep{chernozhukov2018double,dai2022significance,williamson2023general}.

Partition $[n]$ into $K\ge 2$ disjoint folds $I_1,\dots,I_K$. For each $k\in[K]$, let $I_{-k}=[n]\setminus I_k$ denote the training indices and $I_k$ the test indices.
\begin{enumerate}
\item \textbf{Fit on $I_{-k}$.} Using only observations in $I_{-k}$, fit a reduction map $\widehat{\varphi}_{-k}$ and compute surrogate representations $\widehat{\bfR}_i^{(k)}=\widehat{\varphi}_{-k}(\bfX_i)$ for all $i\in[n]$. Learn two prediction rules $\widehat{g}_{-k;d}$ and $\widehat{g}_{-k;d_{\max}}$, by regressing $\bfY_i$ on the first $d$ and the entire $d_{\max}$ coordinates of $\widehat{\bfR}_i^{(k)}$, respectively. The reduction map can be adapted to specific dimension-reduction tasks and methods; see Section \ref{suppsec:reduction map} of the Supplementary Material for detailed examples. The prediction rules may be fit using flexible learner such as kernels, splines, trees, and neural networks.
\item \textbf{Evaluate on $I_k$.} Section \ref{sec:zipper} describes how risks are evaluated only on the held-out fold $I_k$ to form a nondegenerate test statistic.
\end{enumerate}

\subsubsection{A second split to avoid degeneracy}\label{sec:zipper}

For a fitted reduction map $\widehat{\varphi}_{-k}$, define the conditional risk $$\widehat{\V}_{-k}(g,\Q)=\E_\Q\{\ell(\bfY,g(\widehat{\varphi}_{-k}(\bfX)))\mid\widehat{\varphi}_{-k}\}.$$ A naive plug-in estimator of the gap $\V_d-\V_{d_{\max}}$ would evaluate both risks on the same test fold $I_k$:
\[
\widehat{\V}_{-k}(\widehat{g}_{-k;d},\widehat{\Q}_{k})-\widehat{\V}_{-k}(\widehat{g}_{-k;d_{\max}},\widehat{\Q}_k),
\]
where $\widehat{\Q}_k=|I_k|^{-1}\sum_{i\in I_k}\delta_{(\bfX_i,\bfY_i)}$ is the empirical distribution on $I_k$. However, under the null hypothesis $H_{0,d}$, this contrast is \textit{degenerate}: the leading influence-function terms of the two risk evaluations coincide and thus cancel (see \eqref{eq:degeneracy} in the Supplementary Material). Similar degeneracy arises in variable-importance inference \citep[e.g.,][]{dai2022significance,williamson2023general,chen2024zipper}. To obtain a nondegenerate statistic with good power, we further split each test fold using an overlapping scheme inspired by \cite{chen2024zipper}.

Randomly divide $I_k$ into three disjoint parts $I_{k,a}, I_{k,b}, I_{k,o}$ with $|I_{k,a}|=|I_{k,b}|$. The subset $I_{k,o}$ is shared by the two risk evaluations, and we summarize the amount of sharing by the overlap proportion $\tau={|I_{k,o}|}/({|I_{k,a}|+|I_{k,o}|})\in[0,1)$. For $J\in\{a,b,o\}$, let $\widehat{\Q}_{k,J}=|I_{k,J}|^{-1}\sum_{i\in I_{k,J}}\delta_{(\bfX_i,\bfY_i)}$ be the empirical distribution on $I_{k,J}$. Define the fold-specific risk estimates
\begin{align*}
\widehat{\V}_{d,k}&=\tau \widehat{\V}_{-k}(\widehat{g}_{-k;d},\widehat{\Q}_{k,o})+(1-\tau)\widehat{\V}_{-k}(\widehat{g}_{-k;d},\widehat{\Q}_{k,a})\text{ and}\\
\widehat{\V}_{d_{\max},k}&=\tau \widehat{\V}_{-k}(\widehat{g}_{-k;d_{\max}},\widehat{\Q}_{k,o})+(1-\tau)\widehat{\V}_{-k}(\widehat{g}_{-k;d_{\max}},\widehat{\Q}_{k,b}),
\end{align*}
and the cross-fitted, fold-aggregated contrast
\begin{equation}\label{eq:psi}
\widehat{\psi}_{d}=\frac{1}{K}\sum\limits_{k=1}^K(\widehat{\V}_{d,k}-\widehat{\V}_{d_{\max},k}).
\end{equation}

Under $H_{0,d}$, Theorem \ref{thm:thm1} below shows that $$\{n/(2-\tau)\}^{1/2}\widehat{\psi}_{d}\rightsquigarrow \mathcal{N}(0,\nu_{d}^2).$$ To implement the test, we estimate $\nu_{d}^2$ by $\widehat{\nu}_{d}^2=(1-\tau)K^{-1}\sum_{k=1}^K(\widehat{\sigma}^2_{d,k}+\widehat{\sigma}^2_{d_{\max},k}),$ where $\widehat{\sigma}^2_{d,k}=\left|I_{k}\right|^{-1}\sum_{i\in I_k}\{\ell(\bfY_i,\widehat{g}_{-k;d}(\widehat{\bfR}_i^{(k)}))-\widehat{\V}_{-k}(\widehat{g}_{-k;d},\widehat{\Q}_k )\}^2$ and $\widehat{\sigma}^2_{d_{\max},k}$ is defined analogously. Consistency of $\widehat{\nu}_{d}^2$ is established in Proposition \ref{prop:consistency}. We reject $H_{0,d}$ at level $\alpha$ when $$\widehat{T}_d=\{n/(2-\tau)\}^{1/2}\widehat{\psi}_{d}/\widehat{\nu}_d \geq z_{1-\alpha},$$ where $z_{1-\alpha}$ is the $(1-\alpha)$-quantile of $\mathcal{N}(0,1)$.

\begin{remark}
The asymptotic theory permits any fixed $\tau\in[0,1)$, but finite-sample performance reflects an efficiency-degeneracy trade-off. Larger $\tau$ increases overlap, typically reducing variance and improving power; however, as $\tau\to 1$ the contrast becomes nearly degenerate and size can inflate. Following the guidance of \cite{chen2024zipper}, we choose $\tau$ to keep the effective independent test size $(1-\tau)n/(2-\tau)$ comfortably large. In our simulations, we set $\tau = 0.8$ which delivered stable size and power. Data-driven selection of $\tau$ is left for future work.
\end{remark}

\subsection{Sequential selection of the predictiveness-induced order}\label{sec:order selection}

Sections \ref{sec:order}--\ref{sec:estimation-inference} provide the ingredients for our \textit{predictive order determination (POD)} method. POD combines the cross-fitted predictiveness-gap test with a forward search to estimate the predictiveness-induced order $d^*$.

Starting with $d=0$ and increasing $d$ one step at a time, we apply the test sequentially:
\begin{itemize}
\item compute $\widehat{T}_d$;
\item if $\widehat{T}_d < z_{1-\alpha}$ (i.e., $H_{0,d}$ is not rejected), set $\hat{d}=d$ and stop;
\item otherwise set $d\leftarrow d+1$ and continue, up to $d_{\max}$.
\end{itemize}
\noindent Equivalently, 
\begin{equation}\label{equ:hatd}
\widehat{d}=\min\Big\{\min\big\{d:0\leq d \leq d_{\max},\ \widehat{T}_d < z_{1-\alpha}\big\},d_{\max}\Big\}.
\end{equation}
For ease of implementation, Algorithm~\ref{alg:pod} summarizes the POD procedure.

\begin{algorithm}[H]
\label{alg:pod}
\caption{The POD procedure}

\KwIn{Observed data $\{(\bfX_i,\bfY_i)\}_{i=1}^n$; number of folds $K\geq 2$; search limit $d_{\max}$; overlap proportion $\tau\in[0,1)$; significance level $\alpha\in(0,1)$; loss $\ell$; a dimension-reduction method $\mathcal R$; and a prediction-rule class $\mathcal F$.}

Randomly partition $[n]$ into $K$ folds $I_1,\ldots,I_K$.

\For{$k = 1,\ldots,K$}{
Fit a reduction map $\widehat{\varphi}_{-k}:\mathcal{X}\to\mathbb{R}^{d_{\max}}$ by applying $\mathcal{R}$ to $\{(\bfX_i,\bfY_i)\}_{i\in I_{-k}}$, and compute $\widehat{\bfR}_i^{(k)}=\widehat{\varphi}_{-k}(\bfX_i)$ for all $i\in[n]$.

Randomly partition $I_k$ into three disjoint sets $I_{k,a}$, $I_{k,b}$, and $I_{k,o}$ such that $\lvert I_{k,a}\rvert=\lvert I_{k,b}\rvert$ and $\tau=\lvert I_{k,o}\rvert/(\lvert I_{k,a}\rvert+\lvert I_{k,o}\rvert)\in[0,1)$.
}

\For{$d=0,\ldots,d_{\max}$}{
\For{$k=1,\ldots,K$}{
Using $\{(\widehat{\bfR}_i^{(k)},\bfY_i)\}_{i\in I_{-k}}$, fit prediction rules $\widehat{g}_{-k;d}$ (using the first $d$ coordinates of $\widehat{\bfR}_i^{(k)}$) and $\widehat{g}_{-k;d_{\max}}$ (using all $d_{\max}$ coordinates). Optionally, one may select within $\mathcal{F}$ using an internal procedure based only on $I_{-k}$.

Using $\{(\widehat{\bfR}_i^{(k)},\bfY_i)\}_{i\in I_k}$, compute the loss values $\ell(\bfY_i,\widehat{g}_{-k;d}(\widehat{\bfR}_i^{(k)}))$ and $\ell(\bfY_i,\widehat{g}_{-k;d_{\max}}(\widehat{\bfR}_i^{(k)}))$.

Compute $\widehat{\V}_{d,k}$, $\widehat{\V}_{d_{\max},k}$, $\widehat{\sigma}_{d,k}^2$, and $\widehat{\sigma}_{d_{\max},k}^2$ as in Section \ref{sec:zipper}. 
}

Compute $\widehat{\psi}_d$, $\widehat{\nu}_d^2$, and $\widehat{T}_d$ as in Section \ref{sec:zipper}.
}

\KwOut{Reject $H_{0,d}$ if $\widehat{T}_d>z_{1-\alpha}$. \tcp*[r]{$\triangleright$ Sequential tests}

$\widehat{d}=\min\big\{\min\{d:0\leq d \leq d_{\max},\ \widehat{T}_d < z_{1-\alpha}\},d_{\max}\big\}$. \tcp*[r]{$\triangleright$ Sequential selection}
}
\end{algorithm}

The significance level $\alpha$ governs a transparent trade-off between parsimony and potential predictive gain. Smaller $\alpha$ makes rejection harder and tends to stop earlier, yielding a more compact model; larger $\alpha$ lowers the rejection threshold and tends to select a larger order.

Because each step is tested at level $\alpha$, and the procedure stops at the first non-rejection, the probability of overestimating the true order is controlled by the final test:
\begin{equation}\label{eq:FWER}
\Pr(\widehat{d}>d^*)\leq \Pr(\widehat{T}_{d^*}\ge z_{1-\alpha})\to \alpha.
\end{equation}
Thus, the asymptotic probability of overestimation does not exceed $\alpha$. Upper bounds for the underestimation are given in Theorem \ref{thm:thm2}, and Theorem \ref{thm:thm3} establishes consistency of $\widehat{d}$ when $\alpha=\alpha_n\to 0$ at a suitable rate; see Section \ref{sec:theory}.

\section{Theoretical results}\label{sec:theory}

\subsection{Asymptotic null distribution}\label{sec:theory_null}

Recall that $(\bfX,\bfY)\sim\Q$ on $\mathcal{X}\times\mathcal{Y}$ and $\bfZ=(\bfR,\bfX,\bfY)\sim\P$. Let $\widehat{\mathcal{Y}}$ be the prediction space (typically $\widehat{\mathcal{Y}}=\mathcal{Y}$) with norm $\|\cdot\|_{\widehat{\mathcal{Y}}}$ (e.g., the Euclidean norm when $\widehat{\mathcal{Y}}=\mathbb{R}^q$). For a measurable function $h:\text{supp}(\P)\to \widehat{\mathcal{Y}}$ (where $\text{supp}(\P)$ is the support of $\P$), we write $\|h\|$ for a generic norm, such as the $L_2(\P)$-norm $\|h\|_{L_2(\P)}=[ \int \| h(z)\|_{\widehat{\mathcal{Y}}}^2\, d\P(z)]^{1/2}$, or the $L_\infty(\P)$-norm $\|h\|_{L_\infty(\P)}=\sup_{z\in \text{supp}(\P)}\| h(z)\|_{\widehat{\mathcal{Y}}}$. Given a reduction map $\varphi:\mathcal{X}\to\mathbb{R}^{d_{\max}}$ and a prediction rule $g\in\mathcal{G}$, we write $g\circ\varphi$ for the composite function $x\mapsto g(\varphi(x))$.

We study the asymptotic null distribution of the cross-fitted contrast $\widehat{\psi}_d$ in \eqref{eq:psi}, and hence of the test statistic $\widehat{T}_d$. By exchangeability of the $K$ folds, it suffices to analyze a generic term $\widehat{\V}_{-k}(\widehat{g}_{-k;d},\widehat{\Q}_k)$, where $\widehat{\varphi}_{-k}$ and $\widehat{g}_{-k;d}$ are trained on $I_{-k}$ and $\widehat{\Q}_k$ is the empirical distribution on $I_k$. Our goal is to obtain a root-$n$ expansion of $\widehat{\V}_{-k}(\widehat{g}_{-k;d},\widehat{\Q}_k)-\V(g_{d^*},\P)$ under $H_0:\V(g_d,\P)=\V(g_{d^*},\P)$, where $g_d\in{\arg\min}_{g\in\mathcal{G}_d}\V(g,\P)$. To this end, decompose
\begin{align}\label{eq: decomposition}
\widehat{\V}_{-k}(&\widehat{g}_{-k;d},\widehat{\Q}_k)-\V(g_{d^*},\P)\nonumber\\
&=\underbrace{\{\V(g_{d^*},\widehat{\P}_k)-\V(g_{d^*},\P)\}}_{\text{(I)}}+\underbrace{\{\widehat{\V}_{-k}(\widehat{g}_{-k;d},\P)- \V(g_{d^*},\P)\}}_{\text{(II)}}+r_n,
\end{align}
where $\widehat{\P}_k$ is the empirical distribution of $\bfZ$ on $I_k$ (introduced only for theoretical analysis) and $r_n:=\{\widehat{\V}_{-k}(\widehat{g}_{-k;d},\widehat{\P}_k)-\widehat{\V}_{-k}(\widehat{g}_{-k;d},\P)\}-\{\V(g_{d^*},\widehat{\P}_k)-\V(g_{d^*},\P)\}$. Term (I) evaluates the fixed oracle rule $g_{d^*}$ on the test fold and is asymptotically normal by the central limit theorem. The term $r_n$ is a difference-of-differences remainder; cross-fitting typically makes it $o_{\P}(n^{-1/2})$ under mild regularity conditions (cf.\ \citealp{williamson2023general}). Term (II) is the population-risk effect (conditional on the training fold) of replacing $g_{d^*}$ with the learned rule $\widehat{g}_{-k;d}\circ\widehat{\varphi}_{-k}$:
\[
\text{(II)}=\E_\P\big\{\ell\big(\bfY,\widehat{g}_{-k;d}\circ\widehat{\varphi}_{-k}(\bfX)\big)\mid\widehat{g}_{-k;d}\circ\widehat{\varphi}_{-k}\big\} - \E_{\P}\big\{\ell\big(\bfY,g_{d^*}(\bfR)\big)\big\}.
\]
The conditions below ensure that Term (II) is $o_{\P}(n^{-1/2})$.
The mechanism is that $g_{d^*}\in{\arg\min}_{g\in\mathcal{G}_{d^*}}\V(g,\P)$ and that $\bfR$ is sufficient for prediction. For illustration, consider squared loss with scalar response and conditional mean sufficiency, i.e., $\E_{\P}(\bfY\mid \bfX,\bfR)=\E_{\P}(\bfY\mid \bfR)$. Since $g_{d^*}(R)=\E_{\P}(\bfY\mid \bfR)$ minimizes mean squared error, for any $g\in\mathcal{G}_d$ and any $\varphi:\mathcal{X}\to \mathbb{R}^{d_{\max}}$,
\begin{align*}
\E_{\P}\{\ell(\bfY&,g\circ\varphi(\bfX))\} - \E_{\P}\big\{\ell\big(\bfY,g_{d^*}(\bfR)\big)\big\}\\
&=2\E_{\P}\big\{\big(\bfY-g_{d^*}(\bfR)\big)\big(g_{d^*}(\bfR)- g\circ\varphi(\bfX)\big)\big\}+\|g\circ\varphi-g_{d^*}\|_{L_2(\P)}^2.
\end{align*}
The cross term is zero by conditional mean sufficiency, leaving a quadratic remainder. Thus the risk difference is second order in $g\circ\varphi-g_{d^*}$, and Term (II) is negligible if $\|\widehat{g}_{-k;d}\circ\widehat{\varphi}_{-k}-g_{d^*}\|_{L_2(\P)}=o_{\P}(n^{-1/4})$.

We formalize these requirements through three conditions.

\begin{enumerate} 
\item[(C1)](\textit{Risk curvature}) For $d\geq d^*$, there exists a constant $C>0$ such that for any sequence $\{(\varphi^{(j)},g^{(j)})\}_{j\geq 1}$ with $\varphi^{(j)}:\mathcal{X}\rightarrow\mathbb{R}^{d_{\max}}$ and $g^{(j)}\in\mathcal{G}_d$, if $\| g^{(j)} \circ \varphi^{(j)} - g_{d^*} \|\rightarrow 0$, then for all sufficiently large $j$, $\lvert\E_{\P}\{\ell(\bfY,g^{(j)}\circ\varphi^{(j)}(\bfX))-\ell(\bfY,g_{d^*}(\bfR))\}\rvert\leq C \| g^{(j)} \circ \varphi^{(j)} - g_{d^*}\|^2$.
\item[(C2)](\textit{Learning accuracy}) For $d\geq d^*$, $\| \widehat{g}_{-k;d}\circ\widehat{\varphi}_{-k}-g_{d^*}\|=o_{\P}(n^{-1/4})$.
\item[(C3)](\textit{Consistency}) For $d\geq d^*$, $\E_{\P}[\{ \ell(\bfY,\widehat{g}_{-k;d}\circ\widehat{\varphi}_{-k}(\bfX)) -  \ell(\bfY,g_{d^*}(\bfR))\}^2\mid\widehat{g}_{-k;d},\widehat{\varphi}_{-k}] =o_{\P}(1)$.
\end{enumerate}

Condition (C1) requires (i) local curvature of the risk induced by the loss \citep[e.g.,][]{farrell2021deep,williamson2023general} and (ii) a structural link among $(\bfR,\bfX,\bfY)$ that removes first-order terms when the oracle rule is $\bfR$-based but the learned rule is $\bfX$-based. For squared loss, a sufficient condition for (ii) is $\E_{\P}(\bfY\mid\bfX,\bfR)=\E_{\P}(\bfY\mid\bfR)$, which holds in Example~\ref{eg:SDR}. In Example~\ref{eg:factor regression}, it suffices that $\E(\varepsilon\mid\bff,\bfu)=0$, and in Example~\ref{eg:reduced rank regression} that $\E(\varepsilon\mid\bfX)=0$; see, e.g., \citet{yuan2007dimension} and \cite{fan2023factor}. Additional examples (beyond squared loss) and proofs are provided in Section~\ref{subsec: verification of C1} of the Supplementary Material.

Condition (C2) requires that the learned rule converges to $g_{d^*}$ at rate $o_{\P}(n^{-1/4})$. Together with (C1), this implies that Term (II) in \eqref{eq: decomposition} is $o_{\P}(n^{-1/2})$. In many settings, this rate can be achieved when the reduction stage recovers $\bfR^*$ up to rotation and the prediction rule is sufficiently regular. Concretely, under mild conditions, there exists a rotation matrix $O$ such that $\|O\widehat{\bfR}_{[d]}-\bfR^*\|=o_{\P}(n^{-1/4})$. Section~\ref{sec:recovery_null} of the Supplementary Material establishes such a bound in two representative settings: Example~\ref{eg:factor regression}, where $\widehat{\varphi}$ is obtained by principal component analysis, and Example~\ref{eg:SDR}, where $\widehat{\varphi}$ is learned via forward or inverse regression. If $g_{d^*}$ is Lipschitz, this representation error transfers to prediction error. The remaining estimation and approximation errors of $\widehat{g}_{-k;d}$ can be controlled using standard complexity arguments \citep{gyorfi2002distribution}. For example, in factor regression (Example~\ref{eg:factor regression}), \citet{fan2023factor} show that for principal-component reductions with neural-network prediction rules (with an additional split within $I_{-k}$), if $g_{d^*}$ is $\beta$-smooth and $p$ is of the same order as $n$, then $\|\widehat{g}_{-k;d}\circ\widehat{\varphi}_{-k}-g_{d^*}\|_{L_2(\P)}^2=O_{\P}(n^{-2\beta/(2\beta+d^*)})$ up to logarithmic factors, yielding $\|\widehat{g}_{-k;d}\circ\widehat{\varphi}_{-k}-g_{d^*}\|_{L_2(\P)}=o_{\P}(n^{-1/4})$ whenever $\beta>d^*/2$; a closely related rate condition for nonparametric estimation of the outcome regression in average treatment effect inference is given in \cite{fan2025factor}.

Condition (C3) is a mild consistency requirement for the loss, as in \citet{williamson2023general}. For losses that are Lipschitz in their second argument, (C2) typically implies (C3) \citep{farrell2021deep}.

\begin{thm}[Asymptotic null distribution]\label{thm:thm1}
Under Conditions (C1)--(C3) and $H_{0,d}$, for any $\tau\in[0,1)$,
\begin{equation*}
\{n/(2-\tau)\}^{1/2}\widehat{\psi}_{d}\rightsquigarrow\mathcal{N}(0,\nu^2_{d}),
\end{equation*}
where $\nu_{d}^2=2(1-\tau)\sigma^2_{d^*}$ and  $\sigma^2_{d^*}=\mathrm{var}\{\ell(\bfY,g_{d^*}(\bfR))\}$.
\end{thm}

\begin{prop}\label{prop:consistency}
Under Conditions (C1)--(C3) and $H_{0,d}$, for any $\tau\in[0,1)$, $\widehat{\nu}_{d}^2\to\nu_{d}^2$ in probability.
\end{prop}

Theorem \ref{thm:thm1}, Proposition \ref{prop:consistency}, and Slutsky's lemma yield that, under $H_0$, $$\{n/(2-\tau)\}^{1/2}\widehat{\psi}_{d}/\widehat{\nu}_{d}\rightsquigarrow\mathcal{N}(0,1).$$ Consequently, $\Pr(\widehat{d}>d^*)\leq \Pr(\widehat{T}_{d^*}\ge z_{1-\alpha})\to \alpha$, so the sequential procedure controls the asymptotic probability of overestimating the true order at level $\alpha$.

\subsection{Probability bounds for underestimation}

We next control the probability that the sequential selector $\widehat{d}$ in \eqref{equ:hatd} stops too early, i.e., returns a dimension strictly smaller than the predictive order $d^*$. To this end, we study the behavior of the test statistic $\widehat{T}_d$ under the alternative $H_{1,d}$ for $d\in\{0,\dots,d^*-1\}$.

For simplicity, we specialize in this subsection to squared loss. Under $H_{1,d}$ with $d<d^*$, let $\Delta_{d}:=\V_d-\V_{d^*}=\V(g_{d},\P)-\V(g_{d^*},\P)>0$ be the predictive gap. For each fold $k$, under mild conditions, we have
\begin{align*}
\widehat{\V}_{-k}(&\widehat{g}_{-k;d},\widehat{\Q}_k) - \widehat{\V}_{-k}(\widehat{g}_{-k;d^*},\widehat{\Q}_k) \nonumber\\
&=\Delta_{d} + \underbrace{\{\V(g_{d},\widehat{\P}_k)-\V(g_{d^*},\widehat{\P}_k)\}-\Delta_{d}}_{\text{(I')}} + \underbrace{\{\widehat{\V}_{-k}(\widehat{g}_{-k;d},\P)- \V(g_{d},\P)\}}_{\text{(II')}}+o_\P(n^{-1/2}).
\end{align*}
Term (I') is $O_\P(n^{-1/2})$ and is asymptotically normal. Term (II') is the population-risk effect of replacing $g_d$ with $\widehat g_{-k;d}\circ\widehat\varphi_{-k}$, parallel to Term (II); unlike the null regime, it is not necessarily $o_\P(n^{-1/2})$. For squared loss, it can be bounded in terms of the estimation error $\|\widehat{g}_{-k;d}\circ\widehat{\varphi}_{-k}-g_{d}\|_{L_2(\P)}$ together with the deterministic predictive gap $\Delta_d$. The following conditions are analogues of (C2)--(C3) for the regime $d<d^*$.

\begin{enumerate}
\item[(C2')](\textit{Learning accuracy}) For $1\leq d< d^*$, $\|\widehat{g}_{-k;d}\circ\widehat{\varphi}_{-k}-g_{d}\|_{L_2(\P)}=o_{\P}(n^{-1/4})$.

\item[(C3')](\textit{Consistency}) For $1\leq d< d^*$, $\E_{\P}[ \{\|\bfY-\widehat{g}_{-k;d}\circ\widehat{\varphi}_{-k}(\bfX)\|_2^2 -  \|\bfY-g_{d}(\bfR)\|_2^2\}^2\mid\widehat{g}_{-k;d},\widehat{\varphi}_{-k}] =o_{\P}(1)$.
\end{enumerate}

Condition (C2') presumes that, for $d<d^*$, $\bfR_{[d]}$ constitutes a well-posed target for $\widehat{\bfR}_{[d]}$. This identifiability typically follows from eigenvalue separation of the population operator or matrix. In Example \ref{eg:factor regression}, pervasiveness implies that the eigenvalues of $p^{-1}\bfB^{\top}\bfB$ are distinct and bounded away from $0$ and $\infty$; see \cite{fan2013large}. In Example \ref{eg:SDR}, an analogous condition is that the kernel matrix associated with the SDR procedure has distinct leading eigenvalues; see \cite{kyongwon2020post}. Consequently, the Davis-Kahan theorem \citep{davis1970rotation,yu2015useful} guarantees the existence of a rotation matrix $O$ such that $\|O\widehat{\bfR}_{[d]}-\bfR_{[d]}\|=o_{\P}(n^{-1/4})$. This bound in turn supports (C2') under standard regularity conditions and appropriate growth of $n$ and $p$, by the same argument used in the discussion of (C2). Section~\ref{sec:recovery_alter} of the Supplementary Material derives the above bound in two settings: Example~\ref{eg:factor regression}, where $\widehat{\varphi}$ is constructed by principal component analysis, and Example~\ref{eg:SDR}, where $\widehat{\varphi}$ is learned via kernel-matrix-based methods.

\begin{thm}[Underestimation bound]\label{thm:thm2}
Assume Conditions (C1)--(C3) and (C2')--(C3'). Fix $\tau\in[0,1)$. Then, for any $0<\delta<(\V_d-\V_{d^*})$,
\begin{equation*}
\Pr(\widehat{d}<d^*)
\leq \sum_{d=0}^{d^*-1}
\Phi\bigg( \dfrac{\nu_{d}}{\nu_{d,\eta}}z_{1-\alpha}-\dfrac{\{n/(2-\tau)\}^{1/2}(\V_d-\V_{d^*}-\delta)}{\nu_{d,\eta}} \bigg)+o(1),
\end{equation*}
where $\nu_{d}^2=(1-\tau)(\sigma^2_{d}+\sigma_{d^*}^2)$, $\nu_{d,\eta}^2=\nu_{d}^2+\tau\eta_{d}^2$, $\sigma_{d}^2=\mathrm{var}\{\|\bfY-g_{d}(\bfR)\|_2^2\}$, and $\eta_d^2=\mathrm{var}\{\|\bfY-g_{d}(\bfR)\|_2^2-\|\bfY-g_{d^*}(\bfR)\|_2^2 \}$.
\end{thm}
The parameter $\delta>0$ serves as a margin that accounts for the estimation error incurred when replacing $g_d$ with $\widehat{g}_{-k;d}\circ\widehat{\varphi}_{-k}$. Theorem \ref{thm:thm2} further indicates that, asymptotically, larger overlap $\tau$ decreases the probability of underestimating $d^*$.

\subsection{Consistency}

Combining overestimation control \eqref{eq:FWER} with Theorem \ref{thm:thm2} yields a simple consistency criterion. Let the nominal level depend on $n$ and write $\alpha=\alpha_n\to 0$. Theorem \ref{thm:thm1} implies $\Pr(\widehat{d}>d^*)\leq \Pr(\widehat{T}_{d^*}\ge z_{1-\alpha_n})\to 0$. If $z_{1-\alpha_n}/\sqrt{n}\to 0$ (equivalently $(\log\alpha_n)/n\to0$),  the underestimation bound forces $\Pr(\widehat{d}<d^*)\to 0$.

\begin{thm}[Consistency]\label{thm:thm3}
Under the assumptions of Theorems \ref{thm:thm1}--\ref{thm:thm2}, if $\alpha_n\to0$ and $(\log\alpha_n)/n\to0$, then $\lim_{n\to\infty}\Pr(\widehat{d}=d^*)=1$.
\end{thm}

Thus any vanishing sequence $\alpha_{n}$ that decays more slowly than $\exp(-cn)$ guarantees that the sequential selector recovers the oracle order with probability tending to one.

\section{Numerical studies}\label{sec:num}

\subsection{Simulation}\label{sec:simu}

We examine the finite-sample behavior of POD in two canonical settings: factor regression (Example \ref{eg:factor regression}) and sufficient dimension reduction (Example \ref{eg:SDR}). Our goals are to (i) assess the empirical size and power of the individual predictiveness-gap tests, and (ii) evaluate the accuracy of the resulting order estimator.

Each configuration is replicated $500$ times, and results are presented as Monte Carlo averages. Unless stated otherwise, we implement POD with $K=5$ folds, search limit $d_{\max}=8$, and overlap proportion $\tau=0.8$. Preliminary experiments indicated that refitting the reduction map within folds versus fitting it once on the full sample led to negligible differences in the reported metrics; accordingly, we fit $\widehat\varphi$ once using all observations and reuse it across folds to reduce computational cost. Within each training fold, we fit prediction rules from a flexible class $\mathcal{F}$ that includes multiple candidate learners (and, where applicable, tuning parameters). We select within $\mathcal{F}$ using two-fold cross-validation under the loss $\ell$, based only on the training data. The resulting fitted rule is then evaluated on the held-out fold to construct $\widehat T_d$; see Algorithm \ref{alg:pod}. This multi-model class can improve post-reduction prediction by adaptively selecting a strong learner among candidates.

\subsubsection{Factor regression}\label{subsec:simulation for factor}

\paragraph{Data-generating mechanism.} In this experiment, we set the true order $d^*_{\rm F}=5$, representing latent factors $\bff \sim \mathcal{N}(0,I_{d^*_{\rm F}})$ that drive both the response and the covariates. Specifically, $Y=m(\bff)+\varepsilon$ with $m(\bff)=f_1+2f_2+f_3+3f_4+2f_5$ and $\bfX=\bfB\bff+\bfu$, where $\varepsilon\sim\mathcal{N}(0,0.1)$ and $\bfu\sim\mathcal{N}(0, {\rm diag}(v_1^2,\dots,v_p^2))$ are independent noise terms. We take $n=500$ and $p=1,000$. We investigate two different regimes for the loadings and variances:
\begin{itemize}
\item Weak factors: The loadings follow \cite{harding2013estimating}, with variances $v_j=0.55^2$ for $j\in[d^*_{\rm F}]$. Specifically, for each $j$, $b_{lj}=\sqrt{3p^{-1/2}}$ for $l\leq d^*_{\rm F}$ and $b_{lj}=a_{lj}\sqrt{3(p-j)^{-1}}$ for $l>d^*_{\rm F}$, where $a_{lj}=-1$ if $l=rj$ ($r=1,2,\ldots$) and $a_{lj}=1$ otherwise.
\item Pervasive factors: The loadings are i.i.d. from $b_{lj}\sim{\rm Unif}(0,j)$, with variances $v_j^2=25$.
\end{itemize}

\paragraph{Implementation.} The reduction step uses principal component analysis on the sample covariance matrix of $X$. The first $d_{\max}$ eigenvectors are used to form $\widehat{\varphi}$. We take $\mathcal{F}$ to include ordinary least squares (OLS), multivariate adaptive regression splines (MARS, implemented via the \texttt{R} package \texttt{earth} with default settings), and a two-layer neural network (implemented via \texttt{nnet}, with $5$ neurons in the hidden layer). Squared loss is used, so the target order is $d^*(\V)=d^*_{\rm F}$.

\paragraph{Test size and power.} We benchmark POD against two eigenvalue-based sequential tests proposed by \cite{onatski2009testing} and \cite{kapetanios2010testing}. These methods sequentially test $H'_{0,d}:d=d^*_{\rm F}$ versus $H_{1,d}:d<d^*_{\rm F}$ for $d=0,1,\dots,d_{\max}$. The test statistic in \cite{onatski2009testing} is $\max_{d<i\leq d_{\max}}(\widetilde{\lambda}_i-\widetilde{\lambda}_{i+1})/(\widetilde{\lambda}_{i+1}-\widetilde{\lambda}_{i+2})$, where $\widetilde{\lambda}_i$ is the $i$th largest eigenvalue of matrix $\sum_{i=1}^{\lfloor\frac{n}{2}\rfloor}\widetilde{\bfX}_i\widetilde{\bfX}_i^{\top}/\lfloor \frac{n}{2} \rfloor$ with $\widetilde{\bfX}_i=\bfX_i+\sqrt{-1}\bfX_{i+\lfloor \frac{n}{2}\rfloor}$. For \cite{kapetanios2010testing}, the test statistic is $\widehat{\lambda}_{d+1}-\widehat{\lambda}_{d_{\max}+1}$, where $\widehat{\lambda}_i$ is the $i$th largest eigenvalue of the sample covariance matrix. Critical values for both tests are obtained via simulation or subsampling as recommended in the original papers.

\begin{table}[!h]
\renewcommand{\arraystretch}{1.25}
\centering
\caption{Rejection proportions (\%) for each step of the sequential test at $\alpha=5\%$.}
\label{tab:ST factor}
\small
\begin{tabular}{llrrrrr|rr}
\toprule
Scenario & Method     &{$d=0$} &{$d=1$} & {$d=2$} & {$d=3$} & {$d=4$} & {$d=5$} & {$d=6$} \\
\midrule
Weak factors    & POD        & 100     & 100     & 100     & 100     & 100     & 5.4    & 5.4     \\
& \cite{onatski2009testing}    & 22.4      & 24.4      & 25.8    & 29.8    & 27.2    & 3.4   & 3.4    \\
& \cite{kapetanios2010testing}\ & 100     & 100     & 100     & 100     & 100     & 69.4     & 38      \\
\midrule
Pervasive factors    & POD        & 100     & 100     & 100     & 100     & 100     & 6    & 4.8       \\
& \cite{onatski2009testing}\    & 100     & 90.2    & 93      & 96.2    & 50.2    & 5   & 4    \\ 
& \cite{kapetanios2010testing}\ & 100     & 100     & 100     & 100     & 100     & 19.8     & 10.6       \\
\bottomrule
\end{tabular}
\end{table}

Table~\ref{tab:ST factor} reports rejection proportions at $\alpha=5\%$. POD and \cite{onatski2009testing} maintains the nominal size when $d\geq d^*_{\rm F}$; for $d<d^*_{\rm F}$, POD is markedly more powerful. In contrast, \cite{kapetanios2010testing} substantially over-rejects, in line with the size distortions reported in \cite{onatski2009testing}.

\paragraph{Order estimation.} We compare POD with three eigenvalue-based alternatives: the information criteria (IC) from \cite{bai2002determining} (see \eqref{eq:IC} in the Supplementary Material), the eigenvalue ratio (ER) estimator from \cite{ahn2013eigenvalue}, and the adjusted correlation thresholding (ACT) method from \cite{fan2022estimating}.

\begin{figure}[!h]
\centering
\includegraphics[width=.6\linewidth]{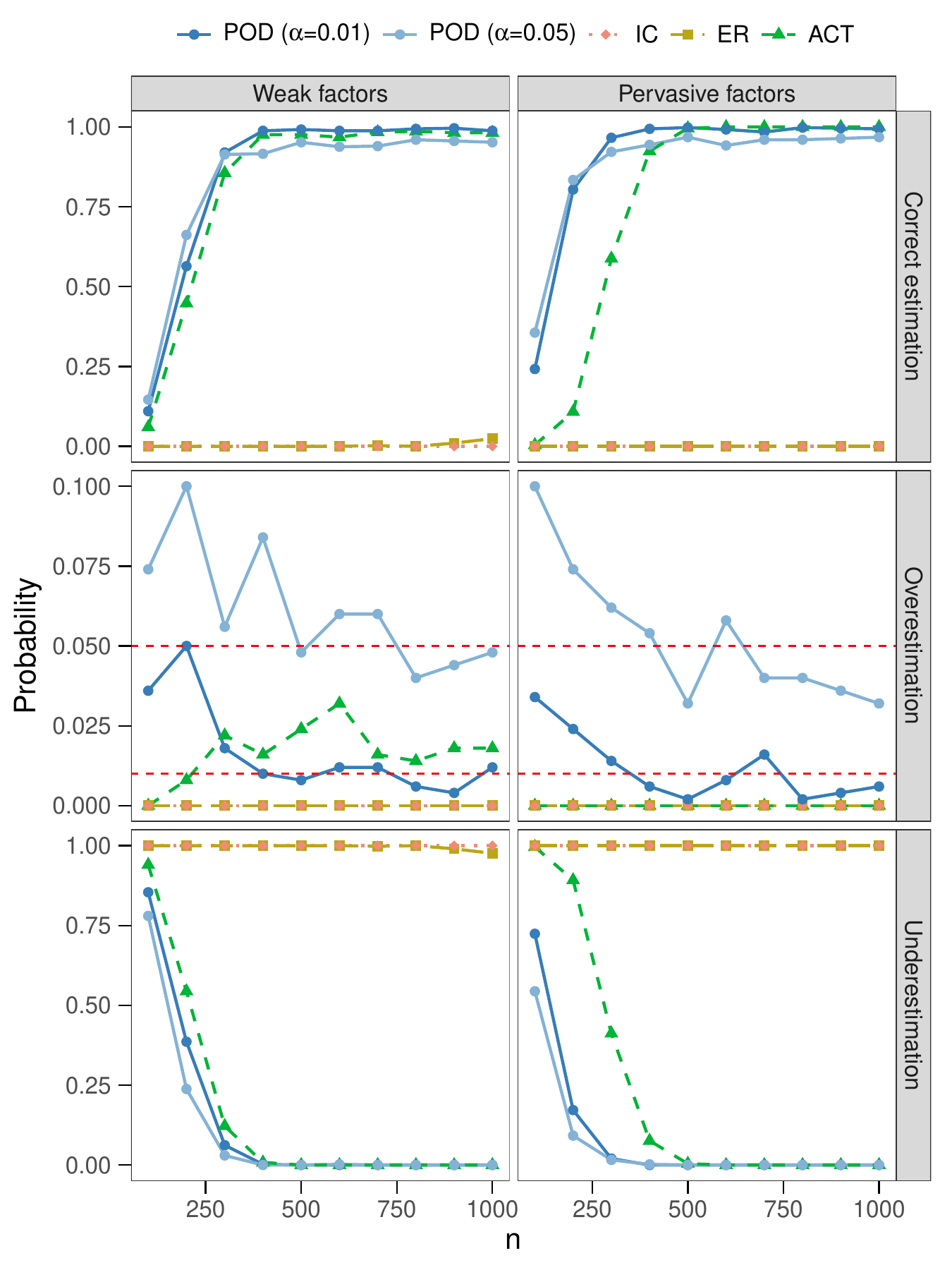}
\caption{Empirical probabilities of correct estimation ($\widehat{d}=d^*_{\rm F}$), overestimation ($\widehat{d}>d^*_{\rm F}$), and underestimation ($\widehat{d}<d^*_{\rm F}$) across $n$ for factor regression. Red dashed horizontal lines in the second row mark the nominal significance levels $\alpha=0.01$ and $\alpha=0.05$.} 
\label{Fig: line_factor}.
\end{figure}

Figure~\ref{Fig: line_factor} shows the probabilities of correct estimation, overestimation, and underestimation of the number of factors as $n$ increases. POD quickly converges to the true order, with its overestimation rate stabilizing at the prescribed $\alpha$ and the underestimation rate approaching zero---consistent with the theory presented in Section \ref{sec:theory}. By comparison, IC and ER exhibit underestimation over the sample sizes considered, consistent with the behavior noted in \cite{fan2022estimating}, whereas ACT converges slowly.

\subsubsection{Sufficient dimension reduction}\label{subsec: simulation for SDR}

\paragraph{Test size and power.} Let $\bfX\sim\mathcal{N}(0,I_p)$ with $p=10$. Consider
\begin{align}
&\text{Model 1}:\qquad Y=X_1+X_2+X_3+X_4+\sigma\varepsilon, \nonumber\\
&\text{Model 2}:\qquad Y=0.4(X_1+X_2+X_3)^2+3\sin[(X_1+X_9+3X_{10})/4]+\sigma\varepsilon, \nonumber
\end{align}
where $\epsilon\sim \mathcal{N}(0,1)$ is independent of $\bfX$ and $\sigma=0.5$. In Model 1, the sample sizes are $n\in\{100,200\}$, and in Model 2, we use $n\in\{200,300\}$. The structural dimensions are $d^*_{\rm CS}=1$ for Model 1 and $d^*_{\rm CS}=2$ for Model 2, with $d^*_{\rm CS}=d^*_{\rm CMS}$ in both cases.

POD uses squared loss to quantify predictiveness, so $d^*(\V)=d^*_{\rm CMS}$. The reduction map is estimated by sliced inverse regression \citep[SIR,][]{li1991sliced} for Model 1, and by directional regression \citep[DR,][]{li2007directional} for Model 2, with $10$ slices for SIR and $4$ slices for DR. We set $\mathcal{F}$ to include OLS, MARS, and regression trees (implemented via the \texttt{R} package \texttt{tree} with default settings).

For benchmarking, we implement two sequential tests proposed by \cite{bura2011dimension} on the same dimension reduction methods as POD for each model. These methods are: (i) a test based on the smallest eigenvalues of the method-specific kernel matrix, with an asymptotic weighted chi-square ($\chi^2$) limit (BY, weighted-$\chi^2$), and (ii) a Wald-type test with an asymptotic $\chi^2$ limit (BY, $\chi^2$). Each procedure tests $H'_{0,d}:d=d^*_{\rm CS}$ versus $H_{1,d}:d<d^*_{\rm CS}$ sequentially for $d=0,1,\ldots,d_{\max}$.

\begin{table}[!h]
\renewcommand{\arraystretch}{1.25}
\centering
\caption{Rejection proportions (\%) for each step of the sequential test at $\alpha=5\%$.}
\label{tab:ST}
{\small
\begin{tabular}{lllrrrrrr}
\toprule
Model                             & $n$                    & Method                      &{$d=0$} &{$d=1$} & {$d=2$} & {$d=3$} & {$d=4$} & {$d=5$}  \\ 
\midrule
1                                  & 100                  & POD                    & \multicolumn{1}{r|}{100}    & 5    & 7.2    & 7.2    & 4.6    & 6.2    \\
&                                     & BY, weighted-$\chi^2$                & \multicolumn{1}{r|}{100}     & 3.4    & 0.2      & 0      & 0      & 0      \\
&                                 & BY, $\chi^2$                    & \multicolumn{1}{r|}{100}    & 100    & 97.4   & 42.8   & 3      & 0      \\ \cline{2-9} 
&                   200                 & POD & \multicolumn{1}{r|}{100}    & 5      & 5.4      & 4.4    & 5.4      & 4.6    
\\ &    & BY, weighted-$\chi^2$ &    \multicolumn{1}{r|}{100}  & 5.2    & 0.2      & 0      & 0      & 0      \\
&    & BY, $\chi^2$ & \multicolumn{1}{r|}{100}     & 82.6   & 22.6   & 1.6    & 0      & 0      \\ 
\midrule
2                                   & 200                  & POD                     & 99.8    & \multicolumn{1}{r|}{96}     & 6.8    & 7      & 7      & 4.8    \\
&                                         & BY, weighted-$\chi^2$                   & 100     & \multicolumn{1}{r|}{52}     & 0.8    & 0.2    & 0      & 0      \\
&                                    & BY, $\chi^2$                      & 100     & \multicolumn{1}{r|}{100}    & 100    & 100    & 100    & 100    \\ \cline{2-9} 
& 300                  & POD                     & 100     & \multicolumn{1}{r|}{99}     & 4.6    & 4.8    & 5.2      & 5.4    \\
&                                          & BY, weighted-$\chi^2$                   & 100     & \multicolumn{1}{r|}{87}    & 1      & 0      & 0      & 0      \\
&   & BY, $\chi^2$                    & 99.8     & \multicolumn{1}{r|}{99.8}    & 80    & 54.4    & 30.6    & 13.4   \\ \bottomrule
\end{tabular}}%
\end{table}

Table~\ref{tab:ST} shows rejection proportions at $\alpha=0.05$ for each step of the sequential test. POD and the BY test with weighted-$\chi^2$ limits maintain size under their respective nulls (POD: $H_{0,d}:d\geq d^*_{\rm CS}$, weighted-$\chi^2$: $H'_{0,d}:d=d^*_{\rm CS}$). In Model 2, POD attains higher power than the weighted-$\chi^2$ test under the alternative $d=1$, while the latter tends to be conservative. In contrast, the Wald-type test with $\chi^2$ limits over-rejects in many configurations, consistent with \cite{bura2011dimension}, who note that substantially larger sample sizes may be required for the $\chi^2$ asymptotic calibration to be accurate.

\paragraph{Order estimation.} Let $\bfX\sim\mathcal{N}(0,I_p)$ with $p=10$. Consider
\begin{align}
&\text{Model 3}:\qquad Y=\sin(X_1)+\sigma\varepsilon, \nonumber\\
&\text{Model 4}:\qquad Y=X_1^2+0.5\sin(X_2)+\sigma\varepsilon, \nonumber\\
&\text{Model 5}:\qquad Y=\lvert X_1\rvert+X_2(X_2+X_3+1)+\sigma\varepsilon,\nonumber
\end{align}
where $\epsilon\sim\mathcal{N}(0,1)$ is independent of $\bfX$ and $\sigma=0.5$. The structural dimensions are $d^*_{\rm CS}=1$ for Model 3 \citep{luo2021order}, $d^*_{\rm CS}=2$ for Model 4, and $d^*_{\rm CS}=3$ for Model 5; in all cases, $d^*_{\rm CS}=d^*_{\rm CMS}$.

POD uses squared-loss predictiveness, hence $d^*(\V)=d^*_{\rm CMS}$. The reduction maps for Models 3--5 are estimated by minimum average variance estimation (MAVE; \citealp{xia2002adaptive}). Prediction rules include OLS, MARS, and regression trees. We compare POD with three alternatives: (i) the ladle estimator \citep{luo2016combining}, which combines the eigenvalue scree with bootstrap variability of the kernel-matrix eigenvectors---to specify kernel, we apply SIR with $10$ slices to Model 3 and DR with $4$ slices to Models~4--5, so that the targeted matrix rank aligns with the oracle $d^*_{\rm CS}$; (ii) a cross-validation (CV) criterion based on MAVE, targeting $d^*_{\rm CMS}$; and (iii) a CV criterion based sliced regression \citep[SR,][]{wang2008sliced}, targeting $d^*_{\rm CS}$.

\begin{figure}[!h]
\centering
\includegraphics[width=0.7\linewidth]{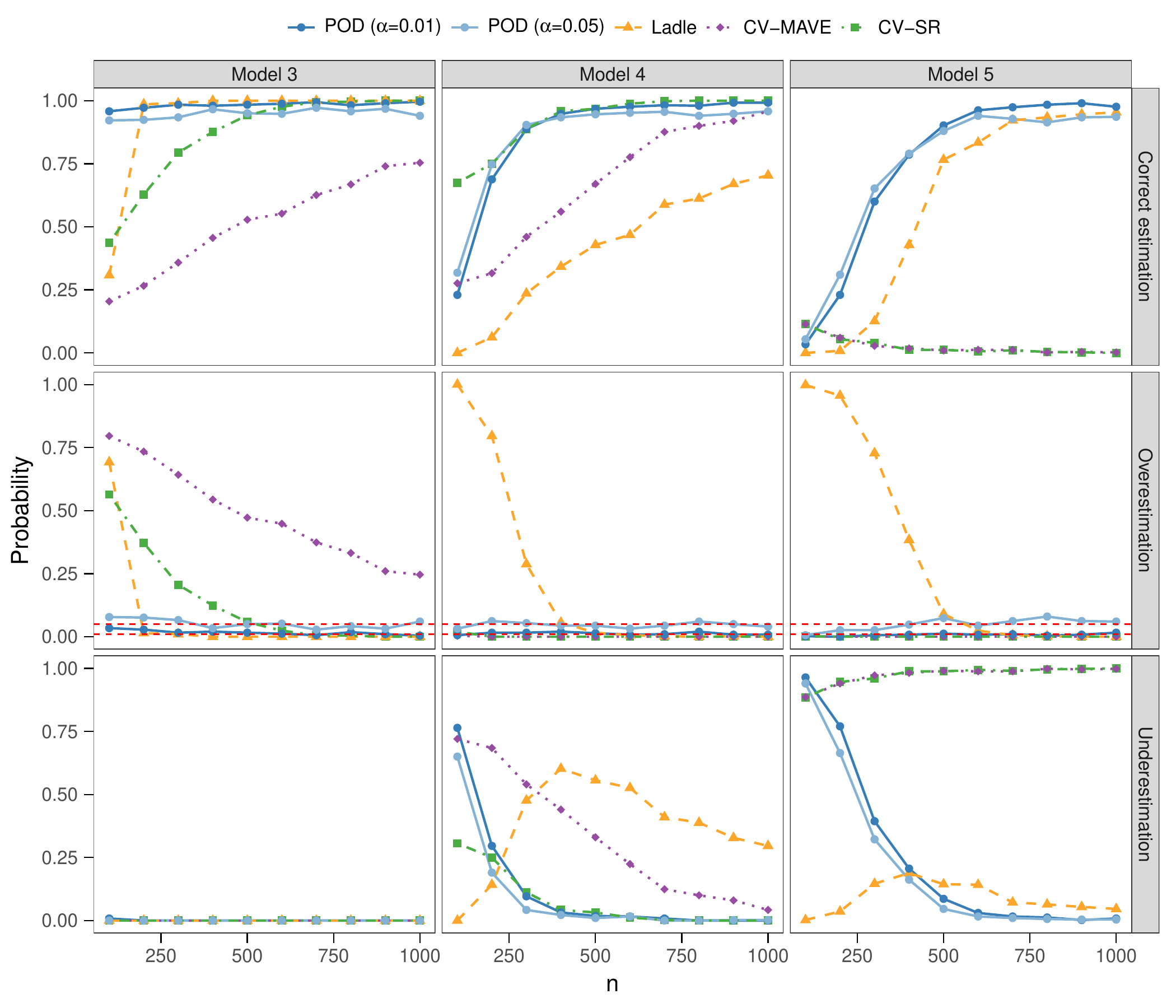}
\caption{Empirical probabilities of correct estimation ($\widehat{d}=d^*_{\rm CMS}$), overestimation ($\widehat{d}>d^*_{\rm CMS}$), and underestimation ($\widehat{d}<d^*_{\rm CMS}$) across $n$ for SDR. Red dashed horizontal lines in the second row mark the nominal significance levels $\alpha=0.01$ and $\alpha=0.05$.} 
\label{Fig: percentages of correct order estimation}.
\end{figure}

Figure~\ref{Fig: percentages of correct order estimation} displays the probabilities of correct estimation, overestimation, and underestimation of $d^*_{\rm CMS}$ as $n$ increases. In all models, POD converges more rapidly than the competitors. Its overestimation probability stabilizes at the nominal level, while the probability of underestimation vanishes. POD can benefit from using a richer candidate class $\mathcal F$, because performance based on a single learner can depend on the data-generating regime. In Section~\ref{sec:ftm simulation} of the Supplementary Material, restricting $\mathcal F$ to OLS only or to regression trees only yields a noticeable loss in accuracy, whereas using MARS alone remains satisfactory; this motivates including multiple candidates in $\mathcal F$.

Results for categorical responses are deferred to Section~\ref{sec:categorical response} of the Supplementary Material.

\subsubsection{Effect of the loss function}

We investigate how the choice of the loss function affects POD through its target dimension $d^*(\V)$. Let $\bfX\sim\mathcal{N}(0,\Sigma)$ with $p=10$ and $\Sigma_{i,j}=(0.5)^{\lvert i-j\rvert}$. Consider a Bernoulli response $Y$ with $\Pr(Y=1\mid X_1>0)=1$ and $\Pr(Y=1\mid X_1\leq 0)=0.6$ \citep{cook2001theory}, so $Y$ depends only on $X_1$. In this setting, $d^*_{\rm CS}=d^*_{\rm CMS}=1$. However, because the Bayes' classifier always predicts $Y=1$, the central discriminant subspace \citep[CDS,][]{cook2001theory} has dimension $d^*_{\rm CDS}=0$. We draw $n=2000$ observations. 

For this problem, under cross-entropy loss, POD targets $d^*(\V)=d^*_{\rm CS}=1$, while under 0-1 loss, it targets $d^*(\V)=d^*_{\rm CDS}=0$. The reduction map is estimated by DR, and the prediction rule is selected from a support vector machine (implemented via the R package \texttt{e1071}, with default settings) and classification trees. We compare POD with three alternatives: the ladle estimator applied to DR, and cross-validation criteria based on MAVE and SR; all methods targets $d^*_{\rm CS}=d^*_{\rm CMS}=1$.

\begin{table}[!h]
\renewcommand{\arraystretch}{1.25}
\centering
\caption{Proportions (\%) of estimated structural dimensions. ave($\widehat{d}$) is the average of the estimates.}
\label{tab:d* under different V}
{\small
\begin{tabular}{lrrcrrrrr}
\toprule
\multicolumn{6}{c}{POD} & Ladle & CV-SR & CV-MAVE \\
& \multicolumn{2}{c}{0-1 loss} & & \multicolumn{2}{c}{Cross-entropy loss} & & & \\ \cline{2-3} \cline{5-6}
& $\alpha=1\%$ & $\alpha=5\%$ & & $\alpha=1\%$ & $\alpha=5\%$ & & & \\
\midrule
{$\widehat{d}=0$}               & 99.2            & 96          &  & 0                  & 0                 & 0     & 0     & 0        \\
{$\widehat{d}=1$}              & 0.8             & 3.4            &  & 99               & 94.2              & 99.4  & 61.2  & 59.4     \\
{$\widehat{d}\geq 2$} & 0               & 0.6           &  & 1                & 5.8               & 0.6   & 38.8  & 40.6     \\
{ave($\widehat{d}$)}            & 0.01            & 0.05           &  & 1.01               & 1.07              & 1.01     & 1.39  & 1.41     \\ \bottomrule
\end{tabular}}%
\end{table}

Table~\ref{tab:d* under different V} shows the proportions of the estimated structural dimensions for different loss functions and methods. The results confirm that POD adapts to the targeted order $d^*(\V)$: under 0-1 loss, POD selects $\widehat{d}=0$, whereas under cross-entropy loss, it identifies $\widehat{d}=1$. In both cases, POD keeps the overestimation rate close to the nominal level. In contrast, the ladle estimator consistently estimates $d^*_{\rm CS}=1$, while SR and MAVE exhibit substantial overestimation.

\subsection{Real-data illustration: PenDigits handwritten digits}\label{sec:app}

We analyze the PenDigits dataset from UCI Machine Learning Repository, focusing on determining how many directional-regression (DR) directions are required to separate the most confusable digits  $\{0,6,9\}$. The training set contains $2,219$ images and the test set $1,035$; each image is represented by a $16$-dimensional pen-trajectory vector. An eigen-analysis of the DR kernel matrix estimated from the training data reveals a sharp drop after the first two eigenvalues, indicating that the first two DR directions capture most of the discriminative information; see Figure~\ref{Fig: real_data}A.

\begin{figure}[!h]
\centering
\includegraphics[width=1\linewidth]{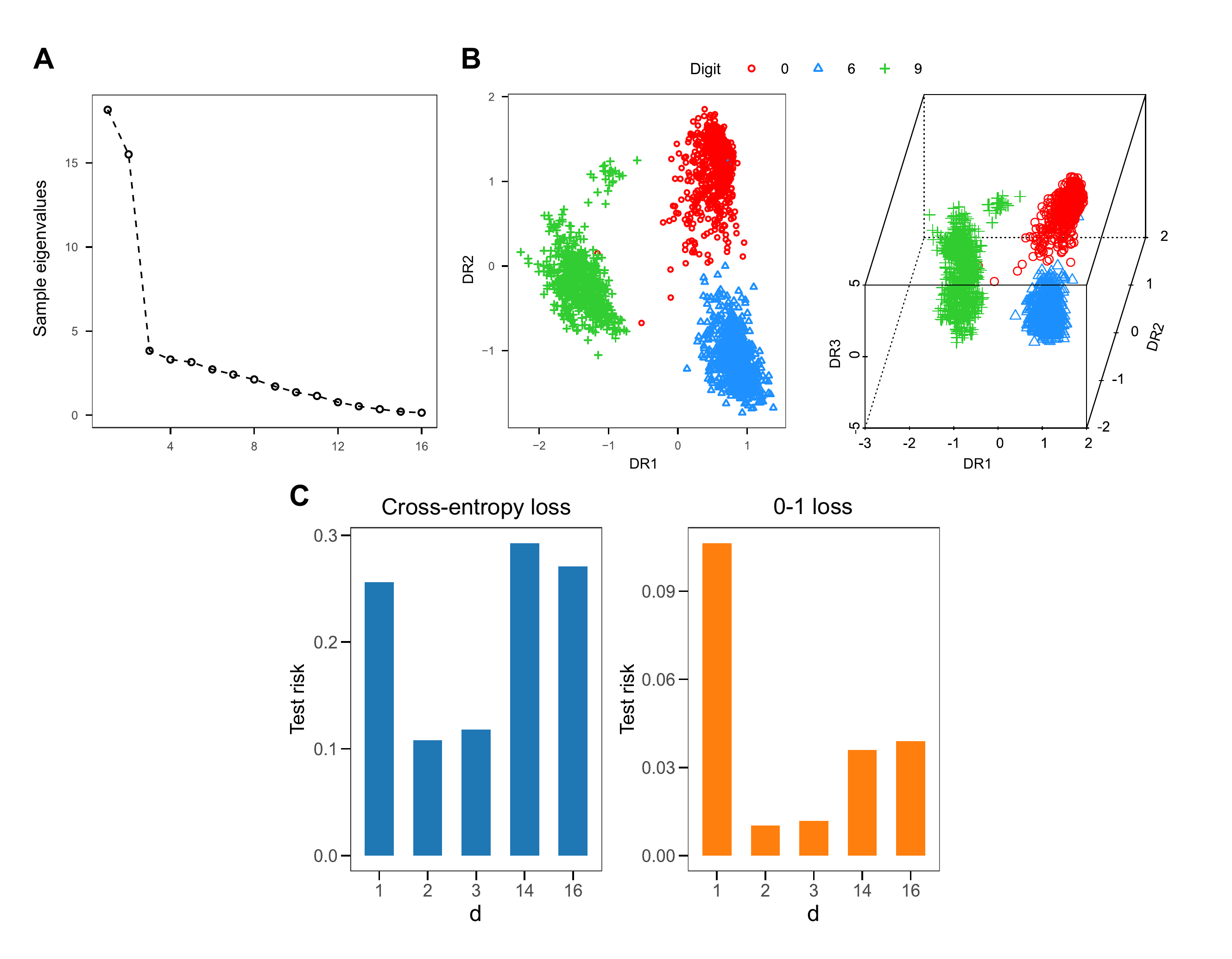}
\caption{PenDigits handwritten digits analysis. Panel A shows the eigenvalue scree plot of the DR kernel matrix. Panel B visualizes the DR representation with the first two and three DR directions. Panel C displays the test risks for different dimensions.}
\label{Fig: real_data}
\end{figure}

Our POD procedure tests the incremental predictive value of the leading DR directions for both cross-entropy and 0-1 loss functions. The prediction rule is a two-layer neural network. Significance levels $\alpha\in\{0.01,0.05\}$ are examined. As benchmarks, we implement the ladle estimator and the BY weighted-$\chi^{2}$ test, both computed from the DR kernel matrix. All methods set the upper search limit $d_{\max}\in\{8,16\}$. To account for the sampling variability introduced by POD's cross-fitting, the ladle's bootstrap, and BY's critical value simulation, each procedure is repeated $100$ times on the fixed training set, and the modal estimate $\widehat{d}$ is reported (see Section \ref{sec:frequency} of the Supplementary Material for frequency details). Across every combination of loss (cross-entropy or 0-1), significance level, and search limit $d_{\max}$, POD consistently selects $\widehat{d}=2$. By contrast, the ladle estimator selects $\widehat{d}=3$ when $d_{\max}=8$, but it inflates the estimate to $\widehat{d}=14$ when $d_{\max}=16$. The weighted-$\chi^2$ test outputs $\widehat{d}=14$ under $\alpha=0.01$, while $\widehat{d}=16$ under $\alpha=0.05$. Figure~\ref{Fig: real_data}B provides a scatter plot of the DR representation: the left panel plots the first two DR directions, while the right panel adds the third. Even in the two-dimensional view, the classes are already almost perfectly separated, confirming that the two DR directions selected by POD capture virtually all discriminative information.

To assess the out-of-sample performance of the models with different dimensions, we proceed as follows. For each candidate dimension $\widehat{d}\in\{1,2,3,14,16\}$, we select the corresponding DR scores, then---under both target loss (cross-entropy or 0-1)---train a classifier on the full training set using a neural network under that same loss function ($500$ repetitions to enhance stability of neural network training). These models are then evaluated on the test set. Figure~\ref{Fig: real_data}C reports the (average) test risk for each chosen dimension. The two-dimensional representation selected by POD yields the smallest test risk for both loss functions. The three-dimensional model performs slightly worse. The larger-dimensional model results in higher test risk, indicating clear overfitting. These results confirm that POD delivers a parsimonious yet highly predictive representation in practice.

\section{Concluding remarks}\label{sec:conclusion}

POD views order determination as a question of predictive adequacy rather than of model-specific structural criteria. It treats the dimension as an index of population risk and selects the smallest order at which additional coordinates no longer improve out-of-sample performance. The method combines cross-fitted risk evaluation, an overlapping second split that avoids degeneracy, and a forward sequential test, and can be applied on top of a wide range of dimension-reduction methods and learners to produce a predictive order and corresponding uncertainty that are aligned with the chosen loss and prediction task.

Several directions merit further investigation. An important direction is robustness, including extensions of POD to dependent or heavy-tailed data. This may call for block cross-fitting schemes for time series or clustered observations, as well as robust loss functions that remain stable under outliers and heavy tails. A second direction is post-selection inference for prediction risk and related functionals at the selected order, developing confidence intervals and tests that account for the data-driven choice of $\widehat{d}$. Together, these extensions would broaden the scope of POD while preserving its central principle: determine dimension by its predictive value.

\par
{\small \baselineskip 10pt
\bibliographystyle{asa}
\bibliography{POD}}

\appendix
\numberwithin{equation}{section}
\renewcommand{\theremark}{\thesection.\arabic{remark}}
\renewcommand{\thethm}{\thesection.\arabic{thm}}
\renewcommand{\theprop}{\thesection.\arabic{prop}}
\renewcommand{\thelemma}{\thesection.\arabic{lemma}}
\setcounter{thm}{0}
\setcounter{prop}{0}
\setcounter{lemma}{0}
\setcounter{remark}{0}
\setcounter{equation}{0}
\beginsupplement
\newpage

\begin{center}
\linespread{2}\selectfont
{\Large \bfseries Supplementary Material for ``Model-Agnostic and Uncertainty-Aware Dimensionality Reduction in Supervised Learning''}
\end{center}
\vspace{30pt}

This Supplementary Material provides supporting details for the main text, including proofs of the theoretical results, additional discussion of assumptions and construction choices, and extended numerical evidence. Section~A collects proofs for the propositions and theorems stated in the paper. Section~B clarifies the predictiveness-induced order $d^*(\V)$ and provides arguments supporting the key conditions used in the analysis, including representation error bounds for the motivating examples. Section~C complements the numerical studies and real-data analysis. Figure~\ref{fig:roadmap} summarizes this organization for quick navigation.

\begin{figure}[!h]
\centering
\includegraphics[
width=\textwidth,
trim=0cm 0cm 6.5cm 0cm,
clip
]{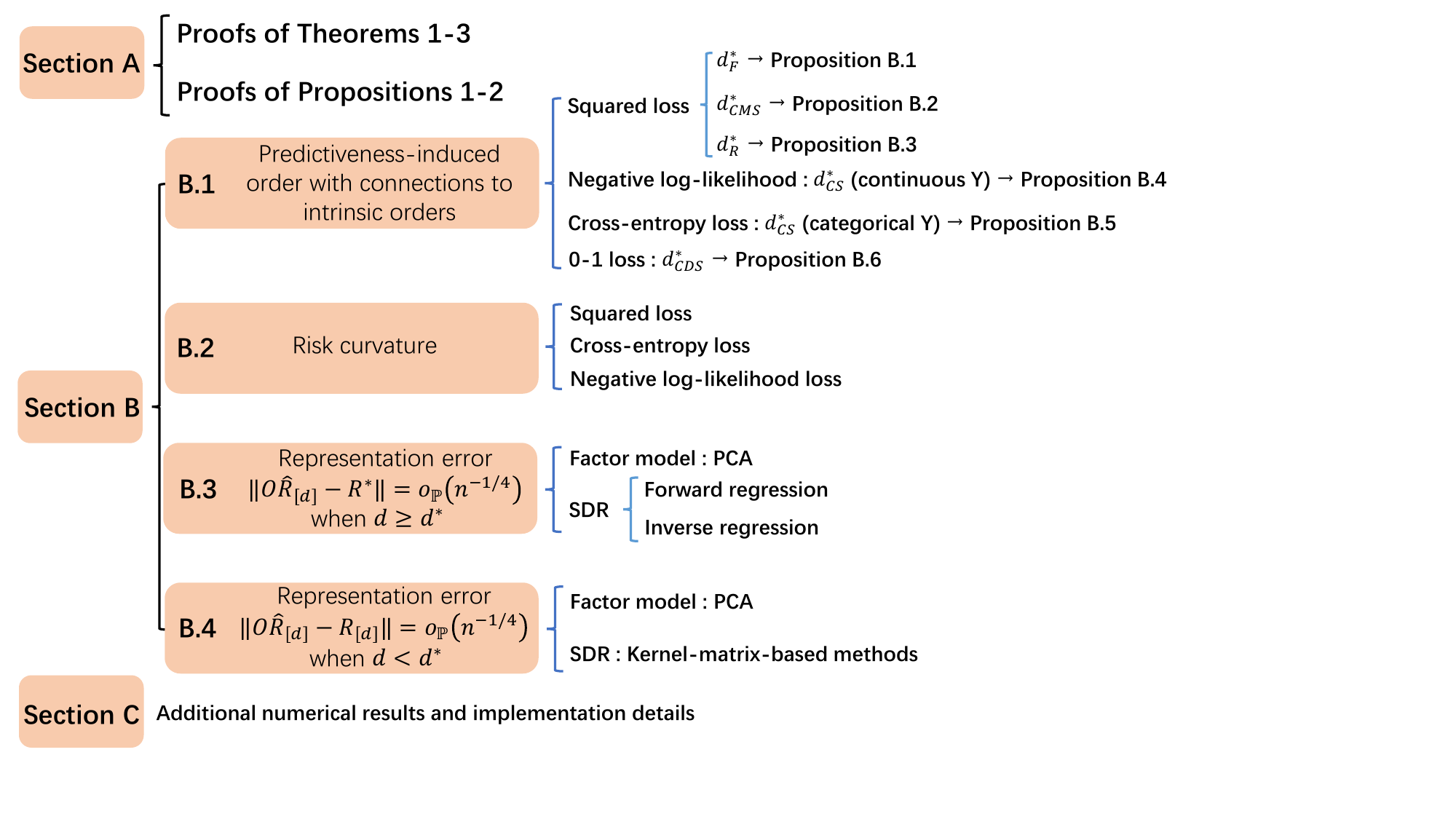}
\caption{Roadmap for the Supplementary Material.}
\label{fig:roadmap}
\end{figure}

\section{Proofs}

\subsection{Proof of Proposition \ref{prop: equivalence  of V}}

\begin{proof}
Recall that in Section \ref{sec:order} we embed $\bfR^*$ into $\mathbb{R}^{d_{\max}}$ by zero padding; that is, $\bfR_{[d^*]}=\bfR^*$ and $\bfR_{-[d^*]}=0$. On the one hand, $\mathcal{G}_0\subseteq\mathcal{G}_1\subseteq\dots\subseteq\mathcal{G}_{d_{\max}}$. Therefore, $\V_0\geq \V_1\geq \dots\geq \V_{d_{\max}}$. On the other hand, consider two cases. \textit{Case 1:} $d^*=0$. Then $\V_d=\min_{g\in\mathcal{G}_d} \E_{\P}\{\ell(\bfY,g(0))\}$. For any $g\in\mathcal{G}_d$, define $h_g(u)=g(0)$. Then $h_g\in\mathcal{G}_0$, and hence $\V_d=\min_{g\in\mathcal{G}_d}\E_{\P}\{\ell(\bfY,h_g(\bfR))\}\ge\min_{g\in\mathcal{G}_0}\E_{\P}\{\ell(\bfY,g(\bfR))\}=\V_0 $. \textit{Case 2:} $d^*\ge 1$. Then $\V_d=\min_{g\in\mathcal{G}_d} \E_{\P}\{\ell(\bfY,g(\bfR^*,0))\}$. For any $g\in\mathcal{G}_d$, define $f_g(\bfu)=g(\bfu_{[d^*]},0)$. Then $f_g\in\mathcal{G}_{d^*}$, and hence, $\V_d=\min_{g\in\mathcal{G}_d} \E_{\P}\{\ell(\bfY,f_g(\bfR))\}\geq\min_{g\in\mathcal{G}_{d^*}} \E_{\P}\{\ell(\bfY,g(\bfR))\}=\V_{d^*}$. Combining these results yields $\V_0\geq \V_1\geq \dots \geq \V_{d^*}=\V_{d^*+1}=\dots=\V_{d_{\max}}$.
\end{proof}

\subsection{Proof of Theorem \ref{thm:thm1}}

\begin{proof}
Recall that 
$\V(g_{d},\P)=\V(g_{d^*},\P),$ for all $d^*\leq d \leq d_{\max}$. Fix $J\in\LRl{o,a,b}$ and $k\in\LRm{K}$. We first decompose 
\begin{equation*}
\widehat{\V}_{-k}(\widehat{g}_{-k;d},\widehat{\P}_{k,J})-\V(g_{d},\P )=\widehat{\V}_{-k}(\widehat{g}_{-k;d},\widehat{\P}_{k,J})-\V(g_{d^*},\widehat{\P}_{k,J})+\V(g_{d^*},\widehat{\P}_{k,J})-\V(g_{d^*},\P ),
\end{equation*}
where
\begin{equation*}
\V(g_{d^*},\widehat{\P}_{k,J})-\V(g_{d^*},\P )=\frac{1}{\left|I_{k,J}\right|}\sum_{i\in I_{k,J}}\ell(Y_i,g_{d^*}(\bfR_i))-\V(g_{d^*},\P)
\end{equation*}
is the leading term for asymptotic normality. Thus, it suffices to show that $\widehat{\V}_{-k}(\widehat{g}_{-k;d},\widehat{\P}_{k,J})-\V(g_{d^*},\widehat{\P}_{k,J})=o_{\P }(n^{-1/2})$.

Write
\begin{align*}
&\widehat{\V}_{-k}(\widehat{g}_{-k;d},\widehat{\P}_{k,J})-\V(g_{d^*},\widehat{\P}_{k,J})\\
= & \underbrace{\widehat{\V}_{-k}(\widehat{g}_{-k;d},\P)-\V(g_{d^*},\P) }_{:=A_n}+\underbrace{\{\widehat{\V}_{-k}(\widehat{g}_{-k;d},\widehat{\P}_{k,J}) - \widehat{\V}_{-k}(\widehat{g}_{-k;d},\P)\}-\{ \V(g_{d^*},\widehat{\P}_{k,J})-\V(g_{d^*},\P) \}  }_{:=B_n}.
\end{align*}

By (C1), $\lvert A_n\rvert\leq C\| \widehat{g}_{-k;d}\circ\widehat{\varphi}_{-k}-g_{d^*}\|^2$. By (C2), $\| \widehat{g}_{-k;d}\circ\widehat{\varphi}_{-k}-g_{d^*}\|=o_{\P}(n^{-1/4})$, so $A_n=o_{\P}(n^{-1/2})$.

Rewrite
\begin{equation*}
B_n=\frac{1}{|I_{k,J}|}\sum_{i\in I_{k,J}}\underbrace{\LRm{\LRl{ \widehat{\V}_{-k}(\widehat{g}_{-k;d},\delta_{\bfZ_i})-\widehat{\V}_{-k}(\widehat{g}_{-k;d},\P) } -\LRl{\V(g_{d^*},\delta_{\bfZ_i})-\V(g_{d^*},\P) }}}_{:=\widehat{\xi}_{d,k}(\bfZ_i)}.
\end{equation*}
Let $\mathcal{D}_k=\LRl{\bfZ_i:i\in I_k}$, for $k=1,\dots,K$. By cross-fitting,  $\E[\widehat{\xi}_{d,k}(\bfZ)\mid\cup_{j\neq k}\mathcal{D}_j]=0$. Therefore, for any $\eta>0$, Chebyshev's inequality yields
\begin{align*}
0 &\leq \Pr (|I_{k,J}|^{-1/2}\sum_{i\in I_{k,J}} \widehat{\xi}_{d,k}(\bfZ_i)>\eta\mid \cup_{j\neq k}\mathcal{D}_j)\\
&\leq \dfrac{\mathrm{var}[ \widehat{\xi}_{d,k}(\bfZ)\mid\cup_{j\neq k}\mathcal{D}_j]}{\eta^2}\leq\dfrac{\E[( \widehat{\xi}_{d,k}(\bfZ))^2\mid \cup_{j\neq k}\mathcal{D}_j]}{\eta^2}.
\end{align*}
By (C1)--(C3),
\begin{align*}
&\E[\{\widehat{\xi}_{d,k}(\bfZ)\}^2\mid\cup_{j\neq k}\mathcal{D}_j ] \\
\leq&2\E[\{\ell(Y,\widehat{g}_{-k:d}\circ \widehat{\varphi}_{-k}(\bfX))-\ell(Y,g_{d^*}(\bfR))\}^2\mid\cup_{j\neq k}\mathcal{D}_j] \\
&\hspace{2em}+2\E[\{\widehat{\V}_{-k}(\widehat{g}_{-k;d},\P)-\V(g_{d^*},\P) \}^2\mid\cup_{j\neq k}\mathcal{D}_j]\\
=&o_{\P}(1).
\end{align*}
Consequently, by dominated convergence, $|I_{k,J}|^{-1/2}\sum_{i\in I_{k,J}} \widehat{\xi}_{d,k}(\bfZ_i)=o_{\P}(1)$, and hence $B_n=o_{\P}(n^{-1/2})$.

Putting the pieces together, under (C1)--(C3), for $ d^*\leq d\leq d_{\max}$, $J\in\LRl{o,a,b}$, and $k\in\LRm{K}$,
\begin{equation}\label{eq:degeneracy}
\widehat{\V}_{-k}(\widehat{g}_{-k;d},\widehat{\P}_{k,J})-\V(g_{d},\P )=\frac{1}{\left|I_{k,J}\right|}\sum_{i\in I_{k,J}} \ell(Y_i,g_{d^*}(\bfR_i))-\V(g_{d^*},\P)+o_{\P }(\left|I_{k,J}\right|^{-1/2}).
\end{equation}
Since $\V_d=\V_{d^*}$ for $d \ge d^*$, it follows that
\begin{align*}
\widehat{\psi}_d
=&\;\frac{1}{K}\sum_{k\in[K]}
\Bigl\{
\tau\cdot \bigl(
\widehat{\V}_{-k}(\widehat{g}_{-k;d},\widehat{\P}_{k,o})
- \widehat{\V}_{-k}(\widehat{g}_{-k;d_{\max}},\widehat{\P}_{k,o})
\bigr) \\
&\hspace{2em}
+ (1-\tau)\cdot \bigl(
\widehat{\V}_{-k}(\widehat{g}_{-k;d},\widehat{\P}_{k,a})
- \widehat{\V}_{-k}(\widehat{g}_{-k;d_{\max}},\widehat{\P}_{k,b})
\bigr)
\Bigr\} \\
=&\;\frac{1}{n/(2-\tau)}
\LRl{
\sum_{k\in[K]}\sum_{i\in I_{k,a}} \ell(Y_i,g_{d^*}(\bfR_i))
-
\sum_{k\in[K]}\sum_{i\in I_{k,b}} \ell(Y_i,g_{d^*}(\bfR_i))
}
+ o_{\P}(n^{-1/2}).
\end{align*}
By the central limit theorem, $$\LRl{n/(2-\tau)}^{1/2}\widehat{\psi}_d\rightsquigarrow\mathcal{N}(0,2(1-\tau)\sigma_{d^*}^2),$$ where $\sigma_{d^*}^2=\mathrm{var}\LRl{ \ell(\bfY,g_{d^*}(\bfR)) }$.
\end{proof}

\subsection{Proof of Proposition \ref{prop:consistency}}

\begin{proof}
Under (C2)--(C3), and by the same arguments used in the proof of Proposition 2.1 and Remark 2.2 of \cite{chen2024zipper}, it follows that for any $d^*\leq d\leq d_{\max}$, $k\in\LRm{K}$, $\widehat{\sigma}_{d,k}^2\to\sigma_{d^*}^2$ in probability. The consistency of $\widehat{\nu}_{d}^2$ then follows from the continuous mapping theorem.
\end{proof} 

\subsection{Proof of Theorem \ref{thm:thm2}}

\begin{proof}
Fix $J\in\LRl{o,a,b}$ and $k\in\LRm{K}$. Decompose $$\widehat{\V}_{-k}(\widehat{g}_{-k;d},\widehat{\P}_{k,J})-\V\LRs{g_{d},\P }=\widehat{\V}_{-k}(\widehat{g}_{-k;d},\widehat{\P}_{k,J})-\V(g_{d},\widehat{\P}_{k,J})+\V(g_{d},\widehat{\P}_{k,J})-\V(g_{d},\P ).$$ Since $\V(g,\P)=\E_{\P}\LRm{\ell(\bfY,g(\bfR))}=\E_{\P}\LRm{\|\bfY-g(\bfR)\|^2}$ is linear in $\P$, we have
\begin{equation*}
\V(g_{d},\widehat{\P}_{k,J})-\V(g_{d},\P )=\frac{1}{\left|I_{k,J}\right|}\sum_{i\in I_{k,J}}\ell(Y_i,g_{d}(\bfR_i))-\V(g_{d},\P).
\end{equation*}
Next, write
\begin{align*}
&\widehat{\V}_{-k}(\widehat{g}_{-k;d},\widehat{\P}_{k,J})-\V(g_{d},\widehat{\P}_{k,J})\\
= & \underbrace{\widehat{\V}_{-k}(\widehat{g}_{-k;d},\P)-\V(g_{d},\P) }_{:=A_n^{\prime}}+\underbrace{\{\widehat{\V}_{-k}(\widehat{g}_{-k;d},\widehat{\P}_{k,J}) - \widehat{\V}_{-k}(\widehat{g}_{-k;d},\P)\}-\{ \V(g_{d},\widehat{\P}_{k,J})-\V(g_{d},\P) \}  }_{:=B_n^{\prime}}.
\end{align*}
We can express
\begin{align*}
A_n^{\prime}&=\E_{\P}\{\|\bfY-\widehat{g}_{-k;d}(\widehat{\varphi}_{-k}(\bfX))\|^2-\|\bfY-g_d(\bfR)\|^2\mid \widehat{g}_{-k;d},\widehat{\varphi}_{-k}\}\\
&=\underbrace{2\E_{\P}\{ \langle\bfY-g_d(\bfR), g_d(\bfR)-\widehat{g}_{-k;d}(\widehat{\varphi}_{-k}(\bfX))\rangle\mid \widehat{g}_{-k;d},\widehat{\varphi}_{-k}\}}_{:=A^{\prime}_{n,k,d}}+\underbrace{\|\widehat{g}_{-k;d}\circ\widehat{\varphi}_{-k}-g_d\|^2_{L_2(\P)}}_{=o_{\P}(n^{-1/2})}.
\end{align*}
By (C2'), $\lvert A_{n,k,d}^{\prime}\rvert\leq 2\sqrt{\mathrm{tr}(\mathrm{var}_{\P}(\bfY))} \|\widehat{g}_{-k;d}\circ\widehat{\varphi}_{-k}-g_d\|_{L_2(\P)}=o_{\P}(n^{-1/4})$. Under (C2')--(C3'), we similarly obtain $B_n^{\prime}=o_{\P}(n^{-1/2})$, by the same argument used to establish $B_n=o_{\P}(n^{-1/2})$ in the proof of Theorem \ref{thm:thm1}.

\textit{Special case:} $d=0$. Under squared loss, set $g_0=\E(\bfY)$, so that $\V(g_0,\P)=\mathrm{tr}(\mathrm{var}_{\P}(\bfY))$. For any $k$, $\widehat{g}_{-k;0}=\bar{\bfY}_{-k}=\sum_{i\in I_{-k}}\bfY_i/\lvert I_{-k}\rvert$, so $\bfY$ is predicted without using $(\bfR,\bfX)$. In this case, 
\begin{equation*}
\widehat{\V}_{-k}(\widehat{g}_{-k;0},\widehat{\P}_{k,J})-\V(g_{0},\P )=\frac{1}{\left|I_{k,J}\right|}\sum_{i\in I_{k,J}}\ell(\bfY_i,g_{0}(\bfR_i))-\V(g_{0},\P)+O_{\P }(n^{-1}).
\end{equation*} 
For notational convenience, set $A_{n,k,0}^{\prime}=0$ for all $k\in[K]$ henceforth.

Consequently, we have the unified expression: under (C2')--(C3'), for $ 0\leq d< d^*$, $J\in\LRl{o,a,b}$, and $k\in\LRm{K}$,
\begin{equation}\label{eq:alter_linear}
\widehat{\V}_{-k}(\widehat{g}_{-k;d},\widehat{\P}_{k,J})-\V(g_{d},\P )=\frac{1}{\left|I_{k,J}\right|}\sum_{i\in I_{k,J}}\ell(\bfY_i,g_{d}(\bfR_i))-\V(g_{d},\P)+A_{n,k,d}^{\prime}+o_{\P }(\left|I_{k,J}\right|^{-1/2}).
\end{equation}

Combining this expansion with the result of
Theorem \ref{thm:thm1} yields, for $0\leq d<d^*$,
\begin{equation*}
\begin{aligned}
&\widehat{\psi}_{d}-(\V_d-\V_{d_{\max}})-\frac{1}{K}\sum_{k=1}^K A_{n,k,d}^{\prime}\\
=&\frac{1}{n/(2-\tau)}\sum_{k\in[K]}\left[\sum_{i\in I_{k,a}}  \{\ell(\bfY_i,g_{d}(\bfR_i))-\V(g_d,\P) \} -\sum_{i\in I_{k,b}}\{\ell(\bfY_i,g_{d^*}(\bfR_i))-\V(g_{d^*},\P) \}\right.\\
&\left.+\sum_{i\in I_{k,o}}\LRl{ \ell(\bfY_i,g_{d}(\bfR_i))-\ell(\bfY_i,g_{d^*}(\bfR_i)) -(\V(g_d,\P)-\V(g_{d^*},\P)) }\right]+o_{\P }\LRs{n^{-1/2}}. 
\end{aligned}
\end{equation*}
By the standard central limit theorem,
\begin{equation*}
\sqrt{n/(2-\tau)}\LRs{\widehat{\psi}_{d}-(\V_d-\V_{d^*})-\frac{1}{K}\sum_{k=1}^K A_{n,k,d}^{\prime}}\rightsquigarrow \mathcal{N}\LRs{0,(1-\tau)\LRs{\sigma_{d}^2+\sigma_{d^*}^2}+\tau\eta_{d}^2},
\end{equation*}
where $\eta_{d}^2=\mathrm{var}_{\P}\{\ell(\bfY,g_{d}(\bfR))-\ell(\bfY,g_{d^*}(\bfR)) \}$. Let $\nu_{d,\eta}^2=(1-\tau)\LRs{\sigma_{d}^2+\sigma_{d^*}^2}+\tau\eta_{d}^2$ and $\nu_{d}^2=(1-\tau)\LRs{\sigma_{d}^2+\sigma_{d^*}^2}$. By an argument analogous to Proposition \ref{prop:consistency}, under (C2')--(C3'), for $0\leq d<d^*$, and any $\tau\in[0,1)$, we have $\widehat{\nu}^2_{d}\to\nu_d^2\text{ in probability}$, as $n\to\infty$. For any $0\leq d<d^*$, $\alpha\in(0,1)$, and sufficiently small $0<\delta<(\V_d-\V_{d^*})$,
\allowdisplaybreaks[2]
\begin{align*}
&\Pr \big(\frac{\sqrt{n/(2- \tau)}\widehat{\psi}_{d}}{\widehat{\nu}_{d}}<z_{1-\alpha}\big)\\
=&\Pr \big(\underbrace{\frac{\sqrt{n/(2- \tau)}\LRs{\widehat{\psi}_{d}-(\V_d-\V_{d^*})-\sum_{k=1}^K A_{n,k,d}^{\prime}/K}}{\nu_{d,\eta}}}_{:=\widehat{\xi}_d}<z_{1-\alpha}\cdot\frac{\widehat{\nu}_{d}}{\nu_{d,\eta}}\\
&\hspace{16em}-\frac{\sqrt{n/(2-\tau)}\big((\V_d-\V_{d^*})+ \sum_{k=1}^K A_{n,k,d}^{\prime}/K\big)}{\nu_{d,\eta}}\big)\\
=&\Pr \big(\widehat{\xi}_d< z_{1-\alpha}\cdot\frac{\widehat{\nu}_{d}}{\nu_{d,\eta}}  -\frac{\sqrt{n/(2-\tau)}\big((\V_d-\V_{d^*})+ \sum_{k=1}^K A_{n,k,d}^{\prime}/K\big)}{\nu_{d,\eta}},\\ &\hspace{16em}\lvert\sum_{k\in[K]} A_{n,k,d}^{\prime}/K\rvert\leq \delta\big)+o(1)\\
\leq &\Pr \big(\widehat{\xi}_d< z_{1-\alpha}\cdot\frac{\widehat{\nu}_{d}}{\nu_{d,\eta}}  -\frac{\sqrt{n/(2-\tau)}(\V_d-\V_{d^*}-\delta)}{\nu_{d,\eta}}\big)+o(1),
\end{align*}
where the second equality holds because $A_{n,k,d}^{\prime}=o_{\P}(1)$ for all $k\in[K]$.

Finally, because the event $\{\widehat{d}<d^*\}=\bigcup_{d=0}^{d^*-1}\{\widehat{T}_d<z_{1-\alpha}\}$, we obtain
\begin{align*}
\Pr (\widehat{d}<d^*)
&\leq \sum_{d=0}^{d^*-1}\Pr \bigg(\frac{\sqrt{n/(2- \tau)}\widehat{\psi}_{d}}{\widehat{\nu}_{d}}<z_{1-\alpha}\bigg)\\
&\leq \sum_{d=0}^{d^*-1}\Pr \bigg(\widehat{\xi}_d< z_{1-\alpha}\cdot\frac{\widehat{\nu}_{d}}{\nu_{d,\eta}}  -\frac{\sqrt{n/(2-\tau)}(\V_d-\V_{d^*}-\delta)}{\nu_{d,\eta}}\bigg)+o(1)\\
&\leq \sum_{d=0}^{d^*-1}\Phi\bigg(\frac{\nu_{d}}{\nu_{d,\eta}}z_{1-\alpha}-\frac{\sqrt{n/(2-\tau)}(\V_d-\V_{d^*}-\delta)}{\nu_{d,\eta}}\bigg)+o(1),
\end{align*}
where the last steps uses $\widehat{\xi}_d\rightsquigarrow \mathcal{N}(0,1)$, $\widehat{\nu}_d^2\to\nu_{d}^2$ in probability, and Polya's theorem \citep[e.g., Lemma 2.11 of][]{van1998asymptotic}.
\end{proof}

\subsection{Proof of Theorem \ref{thm:thm3}}

\begin{proof}
First, by the definition of the testing procedure, 
\begin{align*}
\Pr (\widehat{d}>d^*)&\leq \Pr \bigg(\frac{\sqrt{n/(2- \tau)}\widehat{\psi}_{d^*}}{\widehat{\nu}_{d^*}}\geq z_{1-\alpha_n}\bigg)\\
&\leq\left|\Pr \bigg(\frac{\sqrt{n/(2- \tau)}\widehat{\psi}_{d^*}}{\widehat{\nu}_{d^*}}\geq z_{1-\alpha_n}\bigg)-\alpha_n\right|+\alpha_n\\
&\leq\sup_{\beta\in(0,1)}\left|\Pr \bigg(\frac{\sqrt{n/(2- \tau)}\widehat{\psi}_{d^*}}{\widehat{\nu}_{d^*}}\geq z_{1-\beta}\bigg)-\beta\right|+\alpha_n\\
&=o(1),
\end{align*}
where the last equality follows from Polya's theorem and the assumption $\alpha_n=o(1)$.

Next, by the argument in the proof of Theorem \ref{thm:thm2}, 
\begin{align*}
&\Pr \LRs{\frac{\sqrt{n/(2- \tau)}\widehat{\psi}_{d}}{\widehat{\nu}_{d}}<z_{1-\alpha_n}}\\
=& \Pr \LRs{\frac{\sqrt{n/(2- \tau)}(\widehat{\psi}_{d}-\psi_d)}{\nu_{d,\eta}}-\frac{(\widehat{\nu}_d-\nu_d)}{\nu_{d,\eta}}z_{1-\alpha_n}<\frac{\nu_d}{\nu_{d,\eta}}z_{1-\alpha_n}- \frac{\LRl{n/(2-\tau)}^{1/2}(\V_d-\V_{d^*})}{\nu_{d,\eta}}}.
\end{align*}
Since $\widehat{\nu}_d-\nu_d=O_{\P}(n^{-1/2})$, if $z_{1-\alpha_n}=o(n^{1/2})$ then $$\Pr \LRs{\frac{\sqrt{n/(2- \tau)}\widehat{\psi}_{d}}{\widehat{\nu}_{d}}<z_{1-\alpha_n}}=\Phi\LRs{\frac{\nu_{d}}{\nu_{d,\eta}}z_{1-\alpha_n}-\frac{\LRl{n/(2-\tau)}^{1/2}(\V_d-\V_{d^*}-\delta)}{\nu_{d,\eta}}}+o(1),$$ by Slutsky's theorem and Polya's theorem. Therefore, 
\begin{equation*}
\Pr (\widehat{d}<d^*)\leq \sum_{d=0}^{d^*-1}\Phi\LRs{\frac{\nu_{d}}{\nu_{d,\eta}}z_{1-\alpha_n}-\frac{\LRl{n/(2-\tau)}^{1/2}(\V_d-\V_{d^*}-\delta)}{\nu_{d,\eta}}}+o(1).
\end{equation*}
Finally, recall the standard normal tail approximation:
\begin{equation*}
\Pr\LRl{\mathcal{N}(0,1)>z_{1-\alpha_n}}\sim(2\pi)^{-1/2}z_{1-\alpha_n}^{-1}\exp\LRs{-z_{1-\alpha_n}^2/2},\text{ as } n\to\infty.
\end{equation*}
Let $\gamma_n$ be a sequence such that $\alpha_n=\exp\LRs{-z_{1-\gamma_n}^2/2}$. Then $z_{1-\alpha_n}<z_{1-\gamma_n}=\LRs{-2\log\alpha_n}^{1/2}$. Under the assumptions $\alpha_n\to0$ and $(\log \alpha_n)/n\to0$, we have $z_{1-\alpha_n}=o(n^{1/2})$ and thus $\Pr (\widehat{d}<d^*)=o(1)$ since $(\V_d-\V_{d^*}-\delta)>0$ for $d<d^*$.

Combining these two bounds yields $\Pr (\widehat{d}=d^*)=1$ as $n\to\infty$.
\end{proof}

\section{Predictiveness-induced order and required theoretical conditions}\label{suppsec:details}

\subsection{Predictiveness-induced order with connections to intrinsic orders}\label{subsec:d*_d0}

Unless stated otherwise, we assume $d^*\ge 1$; when $d^*=0$, the definition of $d^*(\V)$ together with Proposition \ref{prop: equivalence of V} implies $d^*(\V)=0$. We begin with the regression, where squared loss and negative log-likelihood are standard choices.

\begin{prop}\label{prop: d*=dF}
In factor regression, under squared loss, $d^*(\V)=d^*_{\rm F}$ provided that the noise $\varepsilon\indep \bff$ has mean zero and finite second moment, and that the regression function satisfies $m(\bff)\neq\E_{\P}\LRs{Y\mid \bff_{S}}$ for every strict subset $S\subsetneq [d^{*}_{\rm F}]$.
\end{prop}

\begin{proof}
For notational consistency, adopt the convention $\sigma(\bff_{[0]})=\sigma(\bff_{\varnothing})=\{\varnothing,\Omega\}$. Suppose, for contradiction, that $d^*(\V)<d^*_{\rm F}$. Then $\E_{\P}(\bfY-m(\bff))^2=\E_{\P}(\bfY-\E_{\P}(\bfY\mid \bff_{[d^*_{\rm F}-1]}))^2$, where $m(\bff)=\E_{\P}(\bfY\mid\bff)$. This implies $m(\bff)=\E_{\P}(\bfY\mid \bff_{[d^*_{\rm F}-1]})$, contradicting the assumption that $m(\bff)\neq\E_{\P}\LRs{Y\mid \bff_{S}}$ for any strict subset $S\subsetneq [d^{*}_{\rm F}]$. Therefore, $d^*(\V)=d^*_{\rm F}$.
\end{proof}

\begin{prop}
In sufficient dimension reduction, under squared loss, $d^*(\V)=d^*_{\rm CMS}$ whenever the central mean subspace exists and $\E_{\P} Y^2<\infty$. 
\end{prop}

\begin{proof}
If $d^*(\V)<d^*_{\rm CMS}$, then by definition $\V_{d^*_{\rm CMS}-1}=\V_{d^*_{\rm CMS}}$, which implies $\E_{\P}(Y\mid \bfX)=\E_{\P}(\bfY\mid\bfB_{d^*_{\rm CMS}-1}^{\top}\bfX)$, where $\bfB_d$ denotes the first $d$ columns of $\bfB$; we adopt the convention $\sigma(\bfB_0^{\top}\bfX)=\{\varnothing,\Omega\}$. This contradicts the minimality of the structural dimension $d^*_{\rm CMS}$. Hence, the conclusion follows.
\end{proof}

\begin{prop}
In reduced-rank regression, under squared loss, $d^*(\V)=d^*_{\rm R}$ provided that the noise $\varepsilon\indep\bfX$ has mean zero and finite second moment.
\end{prop}

\begin{proof}
Suppose, for contradiction, $d^*(\V)<d^*_{\rm R}$. Write the regression function as $f(\bfB^{\top}\bfX)\triangleq\Omega^{\top}\bfB^{\top}\bfX$. Under the stated noise assumptions, $\Omega^{\top}\bfB^{\top}\bfX=\E_{\P}(\bfY\mid\bfB_{d^*_{\rm R}-1}^{\top}\bfX)$. Since $\text{span}(\bfB_{d^*_{\rm R}-1})\subsetneq \text{span}(\bfB)$, there exists $\alpha\in\mathbb{R}^p$ such that $\bfB_{d^*_{\rm R}-1}^{\top}\alpha= 0$ but $\bfB^{\top}\alpha\neq { 0}$. Consequently, for any scalar $k\in\mathbb{R}$, $$ f(\bfB^{\top}\bfX+k\bfB^{\rmT}\alpha)=\E_{\P}(\bfY\mid\bfB_{d^*_{\rm R}-1}^{\top}\bfX+k \bfB_{d^*_{\rm R}-1}^{\rmT}\alpha)=\E_{\P}(\bfY\mid \bfB_{d^*_{\rm R}-1}^{\rmT}\bfX)=f(\bfB^{\rmT}\bfX).$$ Let $\beta_1=\bfB^{\top}\alpha/\|\bfB^{\top}\alpha\|\in\mathbb{R}^{d^*_{\rm R}}$, and extend it to an orthonormal basis $\{\beta_1,\beta_2,\dots,\beta_{d^*_{\rm R}}\}$ of $\mathbb{R}^{d^*_{\rm R}}$. Then $$ f(\bfB^{\rmT}\bfX)=f\big(\sum_{j\in [d^*_{\rm R}] }\langle\bfB^{\rmT}\bfX,\beta_j\rangle\beta_j \big)=f\big(\sum\limits_{j=2}^{d^*_{\rm R}} \langle \bfB^{\rmT}\bfX,\beta_j\rangle\beta_j\big),$$ so $f(\bfB^{\top}\bfX)=f(\widetilde{\bfB}^{\top}\bfX)$, where $\widetilde{\bfB}= \bfB\sum\limits_{j=2}^{d^*_{\rm R}} \beta_j \beta_j^{\top}$ has rank at most $d^*_{\rm R}-1$. Hence, $\Omega^{\top}\bfB^{\top}\bfX=\Omega^{\top}\widetilde{\bfB}^{\top}\bfX$, for all $\bfX$, implying $\bfB\Omega=\widetilde{\bfB}\Omega$. This contradicts the assumption that $A=\bfB\Omega$ has full rank $d^*_{\rm R}$. Therefore, $d^*_{\rm R}(\V)=d^*_{\rm R}$.
\end{proof}

\begin{prop}
Let $\ell(y,g(\cdot\mid r))=-\log g(y\mid r)$ be the negative log-likelihood, where $g(\cdot\mid r)$ denotes the conditional density selected by $\bfR=r$ and lies in the class 
\[
\mathcal G =
\Bigl\{g:\mathbb{R}^{d_{\max}}\to L^1(\mathcal Y)\mid
g(\cdot\mid r)\ge 0,
\int_{\mathcal Y} g(y\mid r)\,d y = 1
\Bigr\}.
\]
Then $d^*(\V)= d^*_{\rm CS}$.
\end{prop}

\begin{proof}
The conditional density $f(\cdot\mid r)$ of $\bfY$ given $\bfR=r$ is the unconstrained population minimizer of the expected negative log-likelihood over $\mathcal{G}$. Suppose, for contradiction, that $d^*(\V)<d^*_{\rm CS}$. Then $\V_{d^*_{\rm CS}}=\V_{d^*_{\rm CS}-1}$, which implies $$\E_{\P}\{\log f(\bfY\mid \bfB_{d^*_{\rm CS}-1}^{\top}\bfX )\}= \E_{\P}\{\log f(\bfY\mid \bfB_{d^*_{\rm CS}}^{\top}\bfX )\}.$$ Thus, $\E_{\P}\bigg[\E_{\P}\Big\{\log \dfrac{f(\bfY\mid \bfB_{d^*_{\rm CS}-1}^{\top}\bfX ) }{f(\bfY\mid \bfB_{d^*_{\rm CS}}^{\top}\bfX )}\mid \bfX\Big\}\bigg]=0.$ By Jensen's inequality, $$\E_{\P}\Big\{\log \dfrac{f(\bfY\mid \bfB_{d^*_{\rm CS}-1}^{\top}\bfX ) }{f(\bfY\mid \bfB_{d^*_{\rm CS}}^{\top}\bfX )}\mid \bfX\Big\}\leq \log \E_{\P}\Big\{ \dfrac{f(\bfY\mid \bfB_{d^*_{\rm CS}-1}^{\top}\bfX ) }{f(\bfY\mid \bfB_{d^*_{\rm CS}}^{\top}\bfX )}\mid\bfX\Big\}=\log1=0,$$ where the equality $\E_{\P}\Big\{ \dfrac{f(\bfY\mid \bfB_{d^*_{\rm CS}-1}^{\top}\bfX ) }{f(\bfY\mid \bfB_{d^*_{\rm CS}}^{\top}\bfX )}\mid\bfX\Big\}=1$ follows from the definition of the central subspace. Thus $$\E_{\P}\Big\{\log \dfrac{f(\bfY\mid \bfB_{d^*_{\rm CS}-1}^{\top}\bfX ) }{f(\bfY\mid \bfB_{d^*_{\rm CS}}^{\top}\bfX )}\mid \bfX\Big\}=0\text{ almost surely.}$$ Equality in Jensen's inequality holds if and only if $f(\bfY\mid \bfB_{d^*_{\rm CS}-1}^{\top}\bfX ) = f(\bfY\mid \bfB_{d^*_{\rm CS}}^{\top}\bfX )$ almost surely, contradicting the definition of $d^*_{\rm CS}$. Therefore, $d^*(\V)=d^*_{\rm CS}.$
\end{proof}

We now turn to classification, where the response $Y$ is categorical. We consider two widely used loss functions: cross-entropy and 0-1 loss.

\begin{prop}
Under cross-entropy loss, $d^*(\V)=d^*_{\rm CS}$ provided that the central subspace exists.
\end{prop}

\begin{proof}
Suppose, for contradiction, that $d^*(\V)<d^*_{\rm CS}$. Then $\V_{d^*_{\rm CS}}=\V_{d^*_{\rm CS}-1}$, which implies $H(Y\mid \bfB_{d^*_{\rm CS}-1}^{\top}\bfX)=H(Y\mid \bfB_{d^*_{\rm CS}}^{\top}\bfX)$, where $H(Y\mid R)=-\E_{\P}\{\log \Pr(Y\mid R)\}$ denotes conditional entropy. By the data-processing inequality for conditional entropy, $$ H(Y\mid \bfB_{d^*_{\rm CS}-1}^{\top}\bfX)\geq H(Y\mid \bfB_{d^*_{\rm CS}}^{\top}\bfX),$$ with equality if and only if $$ Y\indep\bfB_{d^*_{\rm CS}}^{\top}\bfX\mid\bfB_{d^*_{\rm CS}-1}^{\top}\bfX.$$ This implies $$ \Pr (Y\mid \bfB_{d^*_{\rm CS}-1}^{\top}\bfX)=\Pr (Y\mid\bfX),$$ contradicting the definition of $d^*_{\rm CS}$ as the minimal dimension such that $\Pr (Y\mid\bfX)= \Pr (Y\mid \bfB_{d^*_{\rm CS}}^{\top}\bfX)$. Hence, the conclusion follows.
\end{proof}

\begin{prop}
Under 0-1 loss, $d^*(\V)=d^*_{\rm CDS}$ provided that the central discriminant subspace exists, where $d^*_{\rm CDS}$ denotes its structural dimension.
\end{prop}

\begin{proof}
Define the optimal classifiers
\begin{align*}
b_{d^*_{\rm CDS}-1}(\bfX)&=\arg\max_y \Pr (\bfY=y\mid \bfB_{d^*_{\rm CDS}-1}^{\top}\bfX),\\
b_{d^*_{\rm CDS}}(\bfX)&=\arg\max_y \Pr (\bfY=y\mid \bfB_{d^*_{\rm CDS}}^{\top}\bfX).
\end{align*}
By the definition of the central discriminant subspace \citep{cook2001theory}, $ \arg\max_y \Pr (\bfY=y\mid \bfB_{d^*_{\rm CDS}}^{\top}\bfX)=\arg\max_y \eta_y(\bfX)$, where $\eta_y(\bfX)=\Pr (Y=y\mid\bfX) $. Assume $d^*(\V)<d^*_{\rm CDS}$. Define $ \eta_{*}(\bfX)=\max_y \eta_y(\bfX)$. Then $ \V_{d^*_{\rm CDS}}=\V_{d^*_{\rm CDS}-1}$, and under the 0-1 loss,
\begin{align*}
\V_{d^*_{\rm CDS}}&=\Pr (Y\neq b_{d^*_{\rm CDS}}(\bfX) )=1-\E_{\P} (\eta_{*}(\bfX))\\
\V_{d^*_{\rm CDS}-1}&=\Pr (Y\neq b_{d^*_{\rm CDS}-1}(\bfX) )=1-\E_{\P} (\eta_{b_{d^*_{\rm CDS}-1}(\bfX)}(\bfX))
\end{align*}
Thus, $ \E_{\P} (\eta_{*}(\bfX))=\E_{\P} (\eta_{b_{d^*_{\rm CDS}-1}(\bfX)}(\bfX))$. Since $\eta_{*}(\bfX)\geq  \eta_{b_{d^*_{\rm CDS}-1}(\bfX)}(\bfX)$ pointwise, equality of expectations implies $ \eta_{*}(\bfX)= \eta_{b_{d^*_{\rm CDS}-1}(\bfX)}(\bfX),$ almost surely. Hence, $b_{d^*_{\rm CDS}-1}(\bfX)=\argmax_y \eta_y(\bfX)$ almost surely, contradicting the definition of $d^*_{\rm CDS}$. Therefore, $d^*(\V)=d^*_{\rm CDS}$ under 0-1 loss.
\end{proof}

\subsection{Verification of Condition (C1) for common loss functions}\label{subsec: verification of C1}

\begin{enumerate}
\item[(1)] \textbf{Squared loss.}
To verify (C1) under squared loss, we assume that $\bfR$ captures all information in $\bfX$ relevant for predicting the conditional mean of $\bfY$, that is,
\begin{equation}\label{eq: C1 for squared loss}
\E_{\P}(\bfY\mid \bfX,\bfR)=\E_{\P}(\bfY\mid \bfR).   
\end{equation}
For any $g\in\mathcal{G}$ and $\varphi:\mathcal{X}\to \mathbb{R}^{d_{\max}}$, we have 
\begin{align*}
&\lvert\E_{\P}\{\|\bfY-g\circ\varphi(\bfX)\|^2-\|\bfY-g_{d^*}(\bfR)\|^2 \}\rvert\\
=&\|g\circ\varphi(\bfX)-g_{d^*}(\bfR)\|^2_{L_2(\P)}+2\E_{\P}\{\langle \bfY-g_{d^*}(\bfR),g_{d^*}(\bfR)- g\circ\varphi(\bfX)\rangle\}.
\end{align*}
Since $g_{d^*}(r)=\E_{\P}(\bfY\mid \bfR=r)$ is the minimizer of the mean squared error, together with \eqref{eq: C1 for squared loss}, $\E_{\P}\{\langle \bfY-g_{d^*}(\bfR),g_{d^*}(\bfR)- g\circ\varphi(\bfX)\rangle\}=0$ and thus the only contribution comes from $\|g\circ\varphi(\bfX)-g_{d^*}(\bfR)\|^2_{L_2(\P)} $, establishing (C1) under squared loss.

\item[(2)] \textbf{Cross-entropy loss.}
To verify (C1) under cross-entropy loss, we assume the linking condition $\Pr(\bfY\mid \bfX,\bfR)=\Pr(\bfY\mid \bfR)$. For any $g\in\mathcal{G}$ and $\varphi:\mathcal{X}\to \mathbb{R}^{d_{\max}}$, we have 
\begin{align*}
&\Big| \E_{\P}\Big[\sum_{m=1}^M 1\{\bfY=m\} \log\Big(\frac{g\circ\varphi(\bfX)_m}{g_{d^*}(\bfR)_m} \Big) \Big]\Big|\\
=&\Big| \E_{\P}\Big[ \E_{\P}\Big[\sum_{m=1}^M 1\{\bfY=m\} \log\Big(\frac{g\circ\varphi(\bfX)_m}{g_{d^*}(\bfR)_m} \Big)\mid \bfX,\bfR \Big]\Big]\Big|\\
=&\Big| \E_{\P}\Big[\sum_{m=1}^M \Pr(\bfY=m\mid \bfX,\bfR)\log\Big(\frac{g\circ\varphi(\bfX)_m}{g_{d^*}(\bfR)_m} \Big) \Big]\Big|\\
=&\Big| \E_{\P}\Big[\sum_{m=1}^M g_{d^*}(\bfR)_m\log\Big(\frac{g\circ\varphi(\bfX)_m}{g_{d^*}(\bfR)_m} \Big) \Big]\Big|,
\end{align*}
where the last equality uses $g_{d^*}(\bfR)=\Pr(\bfY\mid \bfR)$ and the linking condition. Following the approach in \cite{williamson2023general}, apply a second-order Taylor expansion of $\log(1+t)$ at $t=0$ to obtain, for each class $m\in[M]$, 
\begin{align*}
&\log\Big(\frac{g\circ\varphi(\bfX)_m}{g_{d^*}(\bfR)_m} \Big)\\ =&\log\Big(1+\frac{g\circ\varphi(\bfX)_m-g_{d^*}(\bfR)_m}{g_{d^*}(\bfR)_m} \Big)\\    =&\frac{g\circ\varphi(\bfX)_m-g_{d^*}(\bfR)_m}{g_{d^*}(\bfR)_m}-\frac{(g\circ\varphi(\bfX)_m-g_{d^*}(\bfR)_m)^2}{2g_{d^*}(\bfR)^2_m(1+\xi_m)^2},
\end{align*}
where $\xi_m$ lies pointwise between $0$ and $\frac{g\circ\varphi(\bfX)_m-g_{d^*}(\bfR)_m}{g_{d^*}(\bfR)_m}$. Thus, $g_{d^*}(\bfR)_m(1+\xi_m)$ lies pointwise between $g_{d^*}(\bfR)_m$ and $ g\circ\varphi(\bfX)_m$. Assume that for some $\delta\in(0,1/2)$,  \begin{equation}\label{eq: condition for ce loss}
g_{d^*}(\bfR)_m,g\circ\varphi(\bfX)_m\in(\delta,1-\delta)\text{ almost surely for all }m\in[M].
\end{equation}
Then $$\Big| \E_{\P}\Big[\sum_{m=1}^M g_{d^*}(\bfR)_m\log\Big(\frac{g\circ\varphi(\bfX)_m}{g_{d^*}(\bfR)_m} \Big) \Big]\Big|\leq \frac{1-\delta}{2\delta^2}\|g\circ\varphi-g_{d^*}\|^2_{L_2(\P)},$$ since $\sum_{m=1}^M g\circ\varphi(\bfX)_m=\sum_{m=1}^M g_{d^*}(\bfR)_m=1$. This verifies (C1) for cross-entropy loss under \eqref{eq: condition for ce loss} and the linking condition $\Pr(\bfY\mid \bfX,\bfR)=\Pr(\bfY\mid \bfR)$.

In Example \ref{eg:SDR}, the central subspace is defined by $\Pr(\bfY\mid\bfX)=\Pr(\bfY\mid \bfB^{\top}\bfX)$. Hence, when $\bfR=\bfB^{\top}\bfX$, the linking condition holds automatically.

\item[(3)] \textbf{Negative log-likelihood.} The verification of (C1) for negative log-likelihood is analogous to the cross-entropy case. It suffices to assume the linking condition $$\Pr(\bfY\leq y\mid \bfX,\bfR)=\Pr(\bfY\le y\mid \bfR),\text{ for all }y\in\mathbb{R},$$ together with an analogue of \eqref{eq: condition for ce loss}: for some $\delta\in(0,1/2)$, $$g_{d^*}(y\mid \bfR),g(y\mid \varphi(\bfX))\in(\delta,1-\delta),\text{ for all }y\in\mathbb{R}, \text{ almost surely}.$$

In Example \ref{eg:SDR}, the central subspace is defined by $\bfY\indep\bfX\mid\bfB^{\top}\bfX$. Hence, when $\bfR=\bfB^{\top}\bfX$, the linking condition is again naturally satisfied.
\end{enumerate}

\subsection{Support for Condition (C2): Representation error bounds when \texorpdfstring{$d \ge d^*$}{d >= d*}}\label{sec:recovery_null}

\subsubsection{Factor regression}\label{sec:eg_factor}

Our goal is to verify that, under suitable conditions, the representation error satisfies $\|O \widehat{\bfR}_{[d]}-\bfR^*\|_{L_2(\P)} =o_{\P}(n^{-1/4})$ for some data-dependent rotation matrix $O$. We focus on principal component analysis (PCA) as a representative dimension reduction method for Example \ref{eg:factor regression}. Let $\widehat{\bfv}_1,\dots,\widehat{\bfv}_{d_{\max}}\in\mathbb{R}^p$ denote the leading $d_{\max}$ eigenvectors of the sample covariance matrix computed from $\{\bfX_i:i\in I_{-k}\}$. Define $$\widehat{\bfR}=\widehat{\varphi}_{-k}(\bfX)=p^{-1}\bfW_{d_{\max}}^{\top}\bfX,$$ where $\bfW_{d_{\max}}=\sqrt{p}\LRm{\widehat{\bfv}_1,\dots,\widehat{\bfv}_{d_{\max}}}\in\mathbb{R}^{p\times d_{\max}}$. For brevity, we suppress the fold index $k$ on $\widehat{\bfR}$, since the analysis is identical across folds. Recall from Example \ref{eg:factor regression} and Proposition \ref{prop: d*=dF} that $\bfR^*=\bff$ and $d^*=d^*_{\rm F}$ under mild regularity conditions and squared loss. For $d\ge d^*_{\rm F}$, define $\bfH_d=p^{-1}\bfW_d^{\top}\bfB\in\mathbb{R}^{d\times d^*_{\rm F}}$, where $W_d$ is the first $d$ columns of $\bfW_{d_{\max}}$, and define $O=\bfH_d^{\dagger}\in\mathbb{R}^{d^*_{\rm F}\times d}$ as the pseudo-inverse of $\bfH_d$. The next proposition establishes a representation error bound under additional assumptions.

\begin{description}
\item[\normalfont(D1)](\textit{Boundedness}) Under the factor structure in Example \ref{eg:factor regression}, there exist universal constants $C_1$ and $C_2$ such that:
\begin{itemize}
\item[1.] The factor loading matrix satisfies $\max_{i\in[p],j\in[d_0]}\left|B_{i,j}\right|\leq C_1$.
\item[2.] Each factor $f_i$ is mean zero and bounded in $[-C_2,C_2]$ for all $i\in[d_0]$.
\item[3.] Each idiosyncratic component $u_j$ is mean zero and bounded in $[-C_2,C_2]$ for all $j\in[p]$.
\end{itemize}
\item[\normalfont(D2)](\textit{Weak dependence}) $\sum_{j,j'\in[p],j\neq j'}\left|\E\LRm{u_ju_{j'}}\right|=O(p)$.
\item[\normalfont(D3)](\textit{Bounded projection matrix})
$\max_{i\in[p],j\in[d_{\max}]}\left|(\bfW_{d_{\max}})_{i,j}\right|\leq C_3$.
\item[\normalfont(D4)](\textit{Pervasiveness}) All eigenvalues of $p^{-1}\bfB^{\top}\bf\in\mathbb{R}^{d^*_{\rm F}\times d^*_{\rm F}}$ are bounded away from both $0$ and $\infty$ as $p\to\infty$.
\end{description}

These conditions are standard in the factor-model literature \citep{fan2013large,fan2023factor} and are compatible with (C1)--(C3).

\begin{prop}\label{prop:null_factor}
Let $1\le d^*_{\rm F}\leq d\leq d_{\max}$. Suppose (D1)--(D4) hold and $(\log p)/n=o(1)$ as $n\to\infty$. Then, there exists $O=\bfH_d^{\dagger}\in\mathbb{R}^{d^*_{\rm F}\times d} $ such that \begin{equation*}
\|O\widehat{\bfR}_{\LRm{d}}-\bff\|_{L_2(\P)}=O_{\P }\LRs{p^{-1/2}}.
\end{equation*}
\end{prop}

\begin{proof}
This follows directly from Lemma 2 and Proposition 1 of \cite{fan2023factor}.
\end{proof}

By Proposition \ref{prop:null_factor}, the requirement $\|O \widehat{\bfR}_{[d]}-\bff\|_{L_2(\P)} =o_{\P}(n^{-1/4})$ holds with $O=\bfH_d^{\dagger}$ provided that $\sqrt{n}/p=o(1)$ and $(\log p)/n=o(1)$.

\subsubsection{Sufficient dimension reduction}\label{sec:ex2null}

We show that commonly used sufficient dimension reduction (SDR) methods can satisfy the representation requirement $\|O \widehat{\bfR}-\bfR^*\|=o_{\P}(n^{-1/4})$ for some data-dependent rotation matrix $O$. Assume that $\bfB^{\top}\bfX$ contains all the predictive information in $\bfX$ for the response $Y\in\mathbb{R}$ under the user-chosen predictiveness measure $\V$. Depending on the loss function $\ell$, the span of $\bfB$ may represent the central subspace, the central mean subspace, or another task-specific sufficient subspace, and $d^*$ denotes its structural dimension; see Section \ref{subsec:d*_d0}. For each fold $k$, apply a user-specified SDR method to $\{(\bfX_i,Y_i):i\in I_{-k}\}$ to estimate the target subspace, yielding a basis matrix $\widehat{\bfB}_{d_{\max}}\in\mathbb{R}^{p\times d_{\max}}$ for the estimated $d_{\max}$-dimension subspace. Define $\widehat{\bfR}=\widehat{\bfB}_{d_{\max}}^{\top}\bfX\in\mathbb{R}^{d_{\max}}$, and suppress the fold index for brevity. Recall from Example \ref{eg:SDR} that $\bfR^*=\bfB^{\top}\bfX$. For $d\geq d^*$, define $O=\bfB^{\top}\widehat{\bfB}_d\in\mathbb{R}^{d^*\times d}$, where $\widehat{\bfB}_d$ is the first $d$ columns of $\widehat{B}_{d_{\max}}$. The next proposition characterizes the representation error. 

\begin{prop}\label{prop:null_SDR}
Let $1\le d^*\leq d\leq d_{\max}$. With $O=\bfB^{\top}\widehat{\bfB}_d\in\mathbb{R}^{d^*\times d}$, we have
\begin{equation}\label{eq:SDR_basic}
\|O\widehat{\bfR}_{\LRm{d}}-\bfB^{\top}\bfX\|\leq \|P_{\bfB}-P_{\widehat{\bfB}_{d^*}}\|_2 \|\bfX\|,
\end{equation}
where $P_A$ denotes the orthogonal projection onto the column space of $A$.
\end{prop}

\begin{proof}
Substituting $O=\bfB^{\top}\widehat{\bfB}_d$ and $\widehat{\bfR}_{[d]}=\widehat{\bfB}_d^{\top}\bfX$ into the left-hand side of \eqref{eq:SDR_basic} yields 
\begin{align*}
\|O\widehat{\bfR}_{\LRm{d}}-\bfB^{\top}\bfX\|
=&\|\bfB^{\top}\widehat{\bfB}_d\widehat{\bfB}_d^{\top}\bfX-\bfB^{\top}\bfX\|\\
=&\|\bfB^{\top}(I-P_{\widehat{\bfB}_d})\bfX \|\\
\leq & \|\bfB^{\top}(I-P_{\widehat{\bfB}_d})\|_2\|\bfX \|\\
\leq &  \|\bfB^{\top}(I-P_{\widehat{\bfB}_{d^*}})\|_2\|\bfX \|\\
=& \|P_{\bfB}-P_{\widehat{\bfB}_{d^*}}\|_2\|\bfX \|.
\end{align*}
The second equality uses $P_{\widehat{\bfB}_d}=\widehat{\bfB}_d\widehat{\bfB}_d^{\top}$ since $\widehat{\bfB}_d$ has orthonormal columns. The second inequality holds because $\mathrm{span}({\widehat{\bfB}_{d^*}})\subseteq\mathrm{span}({\widehat{\bfB}_{d}}$) when $d\ge d^*$, where $\mathrm{span}(A)$ denotes the column space of $A$. The final equality follows from the identity $\|A^{\top}Z\|_2=\|P_{A}Z\|_2$ for any $A$ with orthonormal columns and any matrix $Z$, together with Theorem 2.5.1 of \cite{golub2013matrix}.  
\end{proof}

Proposition \ref{prop:null_SDR} shows that, when $\|\bfX\|$ is bounded, the representation error $\|O\widehat{\bfR}_{\LRm{d}}-\bfB^{\top}\bfX\| $ is controlled by either $\|P_{\bfB}-P_{\widehat{\bfB}_{d^*}}\|_2$ or $\|P_{\bfB}-P_{\widehat{\bfB}_{d^*}}\|_F$, which are topologically equivalent \citep{yin2011sufficient}. Below, we summarize known convergence rates for several common SDR methods. These rates indicate that the requirement $\|O \widehat{\bfR}-\bfR^*\|=o_{\P}(n^{-1/4})$ is mild and is typically satisfied by most SDR methods.

\textbf{Forward regression.} 
This class consists of two canonical approaches---the outer product of gradients and the minimum average variance estimator---together with many extensions. See, for example, \cite{xia2002adaptive,hristache2001structure,xia2007constructive,wang2008sliced,yin2011sufficient}. Under the regularity conditions in these works, it is shown that $\widehat{\bfB}_{d^*}$ is consistent with
\begin{equation*}
\|P_{\widehat{\bfB}_{d^*}}-P_{\bfB}\|_2=O_{\P }(h^4+(nh^{d^*})^{-1}\log n+n^{-1/2}),
\end{equation*}
where $h\propto n^{-1/(d^*+4)}$ \citep{xia2002adaptive,wang2008sliced,yin2011sufficient}. Equivalently, $\|P_{\widehat{\bfB}^{d^*}}-P_{\bfB}\|= O_{\P }(r(d^*,n))$, where $$r(d^*,n)=
\begin{cases}
n^{-1/2},&d^*\leq3,\\
(nh^{d^*})^{-1}\log n,&d^*>3.
\end{cases}$$
Thus, the recovery requirement holds whenever $r(d^*,n)=o(n^{-1/4})$. In typical applications with $d^*\leq3$, this condition is automatic, and forward-regression estimators $\widehat{\bfB}_d$ with $d\geq d^*$ achieve the desired representation error when $\|\bfX\|$ is bounded.

\textbf{Inverse regression.}
Representative examples include sliced inverse regression (SIR) \citep{li1991sliced}, sliced average variance estimate \citep{cook1991sliced}, and directional regression (DR) \citep{li2007directional}. These methods typically construct a matrix-valued statistic $\widehat{M}$ that converges to a population matrix $M$ satisfying $P_M=P_{\bfB_{d^*}}$.

When $p<n$ is fixed, prior work \citep{hsing1992asymptotic,zhu1995asymptotics,li2007directional} establishes that $n^{1/2}(\widehat{M}-M)$ converges in distribution to a multivariate normal distribution under mild conditions. Let $\widehat{\bfB}_{d}$ be the eigenvectors corresponding to the largest $d$ eigenvalues of $\widehat{M}$. If the $d^*$th eigenvalue of $M$ is bounded away from 0, then the identity $\|P_{\widehat{\bfB}_{d^*}}-P_{\bfB}\|_F=2^{1/2}\|\sin\Theta(\widehat{\bfB}_{d^*},\bfB)\|_F$, together with the Davis-Kahan sin $\theta$ theorem \citep{davis1970rotation,yu2015useful} and $\widehat{M}-M=O_{\P }(n^{-1/2})$, yields $$\|P_{\widehat{\bfB}_{d^*}}-P_{\bfB}\|_F=O_{\P }(n^{-1/2}).$$ Hence, $\|O \widehat{\bfR}-\bfR^*\|=o_{\P}(n^{-1/4})$ follows immediately.

When $p>n$ and diverges with $n$, additional structure such as sparsity is typically required. Under such conditions, \cite{lin2019sparse} show that the lasso-SIR estimator satisfies $$\|P_{\widehat{\bfB}_{d^*}}-P_{\bfB}\|_F=O_{\P }(\{(s\log p)/(n\lambda)\}^{1/2}),$$ where $s$ is the number of active variables, $\lambda$ is the $d^*$th largest eigenvalue of $\mathrm{var}_{\P}\{\E_{\P} (\bfX\mid Y)\}$, and $n\lambda=p^{\gamma}$ for some $\gamma>1/2$. The representation requirement holds whenever $(s\log p)/(n\lambda)\ll n^{-1/2}$.

\subsection{Support for Condition (C2'): Representation error bounds when \texorpdfstring{$d < d^*$}{d < d*}}\label{sec:recovery_alter}

\subsubsection{Factor regression}\label{sec:alter_ex1}

Our goal is to show that, under suitable conditions, the representation error satisfies $\|O\widehat{\bfR}_{[d]}-\bff_{[d]}\|_{L_2(\P)}=o_{\P}(n^{-1/4})$, for some data-dependent rotation matrix $O$. As discussed in Section \ref{sec:eg_factor}, we use PCA to estimate the latent factors $\bff$. We retain the notation of Section \ref{sec:eg_factor} and focus on the case $d<d^*_{\rm F}$. Since $\widehat{\bfR}=p^{-1}\bfW_{d_{\max}}^{\top}\bfX$, we have $\widehat{\bfR}_{[d]}=p^{-1}\bfW_d^{\top}\bfX$. Without loss of generality, assume that the columns of $\bfB=(\beta_1,\dots,\beta_{d_{\rm F}^*})$ are orthogonal and ordered so that $\|\beta_1\|_2\ge \cdots \geq \|\beta_{d^*_{\rm F}}\|_2$. Let $\bfB_{d}=(\beta_1,\dots,\beta_{d})$ and $\bfB_{-d}=(\beta_{d+1},\dots,\beta_{d^*_{\rm F}})$. For $1\leq d< d^*_{\rm F}$, define $\bfH_{d}=p^{-1}\bfW_d^{\top}\bfB_d\in\mathbb{R}^{d\times d}$ and define $O=\bfH_d^{-1}\in\mathbb{R}^{d\times d}$ as the inverse of $\bfH_d$. The next proposition establishes a representation error bound under (D1)--(D3) and (D4'), replacing (D4) by a strengthened form to address an identifiability issue.

\begin{description}
\item[\normalfont(D4')]  (\textit{Pervasiveness}) All eigenvalues of $p^{-1}\bfB^{\top}\bfB\in\mathbb{R}^{d^*_{\rm F}\times d^*_{\rm F}}$ are distinct and bounded away from both $0$ and $\infty$ as $p\to\infty$.
\end{description}

\begin{prop}\label{prop:dra_factor}
Let $1\leq d<d^*_{\rm F}$. Suppose (D1)--(D3) and (D4') hold, and $n^{-1}\log p=o(1)$ as $n\to\infty$. Then there exists $O=\bfH_d^{-1}\in\mathbb{R}^{d\times d}$ such that $$\|O \widehat{\bfR}_{[d]}-\bff_{[d]}\|_{L_2(\P)}=O_{\P }\LRs{\delta},$$ where $\delta=n^{-1/2}+(n/\log p)^{-1/2}+p^{-1/2}$.
\end{prop}

We first state a technical lemma ensuring that $\nu_{\min}\LRs{\bfH_d}$ is bounded away from 0 in probability. 

\begin{lemma}\label{lemma:lemma4}
Suppose (D1)--(D3) and (D4') hold, and $n^{-1}\log p=o(1)$ as $n\to\infty$. Then for any $\varepsilon>0$ there exists an integer $N$ and a universal constant $C>0$, independent of $n,p,d,d_{\max}$, such that $$\Pr \LRs{\nu_{\min}\LRs{\bfH_d}>C}\geq1-\varepsilon, \quad \text{for all } n \geq N.$$
\end{lemma}

\begin{proof}[Proof of Lemma \ref{lemma:lemma4}]
Under (D1)--(D3) and (D4'), the claim follows by a straightforward extension of the argument used in Proposition 1 of \cite{fan2023factor}.
\end{proof}

\begin{proof}[Proof of Proposition \ref{prop:dra_factor}]
Write
\begin{align*}
\widehat{\bfR}_{\LRm{d}}&=p^{-1}\bfW_d^{\top}\bfX\\
&=p^{-1}\bfW_d^{\top}\LRs{\bfB\bff+\bfu}\\
&=p^{-1}\bfW_d^{\top}\bfB_{d}\bff_{\LRm{d}}+\frac{1}{p}\bfW_d^{\top}\bfB_{-d}\bff_{-\LRm{d}}+p^{-1}\bfW_d^{\top}\bfu\\
&=\bfH_d\bff_{\LRm{d}}+p^{-1}\bfW_d^{\top}\bfB_{-d}\bff_{-\LRm{d}}+p^{-1}\bfW_d^{\top}\bfu.
\end{align*}
Therefore,
\begin{equation}\label{eq:first_term}
\bfH_d^{-1}\widehat{\bfR}_{\LRm{d}}-\bff_{\LRm{d}}=p^{-1}\bfH_{d}^{-1}\bfW_d^{\top}\bfB_{-d}\bff_{-\LRm{d}}+p^{-1}\bfH_d^{-1}\bfW_d^{\top}\bfu.
\end{equation}
Let us begin with the first term on the right-hand side of \eqref{eq:first_term}. We have $$\|p^{-1}\bfH_{d}^{-1}\bfW_d^{\top}\bfB_{-d}\bff_{-\LRm{d}}\|_2\leq p^{-1}\|\bfH_{d}^{-1}\|_2\|\bfW_d^{\top}\bfB_{-d}\|_2\|\bff_{-\LRm{d}}\|_2.$$ Note that $\|\bfH_d^{-1}\|_2=\nu^{-1}_{\min}(\bfH_d)=O_{\P }(1)$ by Lemma \ref{lemma:lemma4}. Moreover, (D1) implies $\E(\|\bff_{-\LRm{d}}\|_2)\leq C$ for some universal constant $C$. It remains to control $\|\bfW_d^{\top}\bfB_{-d}\|_2$, Let $\bfB^{(u)}=\bfB\cdot\text{diag}\LRs{\|\bfbeta_1\|^{-1},\dots,\|\bfbeta_{d_0}\|^{-1}}$ so that $\bfB^{(u)}$ has orthonormal columns, and define $\bfB_{d}^{(u)}$ and $\bfB_{-d}^{(u)}$ analogously. Let $\widehat{\bfW}_d=\bfW_d/\sqrt{p}=\LRm{\widehat{\bfv}_1,\dots,\widehat{\bfv}_{d}}$. Then $\|\bfW_d^{\top}\bfB_{-d}\|_2\leq\sqrt{p}\|\widehat{\bfW}_d^{\top}\bfB_{-d}^{(u)}\|_2\|\bfbeta_{d+1}\|_2$. By singular value decomposition, there exist orthogonal matrices $O_1,O_2\in\mathbb{R}^{d\times d}$ such that $$O_1^{\top}\widehat{\bfW}_d^{\top}\bfB^{(u)}_{d}O_2=\cos\Theta\big(\widehat{\bfW}_d,\bfB^{(u)}_{d}\big).$$ Hence, $\|\widehat{\bfW}_d^{\top}\bfB_{-d}^{(u)}\|_2=\|(\widehat{\bfW}_d-\bfB^{(u)}_{d}O_2O_1^{\top})^{\top}\bfB^{(u)}_{-d}\|_2\leq \|\widehat{\bfW}_d-\bfB^{(u)}_{d}O_2O_1^{\top}\|_2\|\bfB^{(u)}_{-d}\|_2$ since $\bfB$ has orthogonal columns. Moreover, $\|\bfB^{(u)}_{-d}\|_2=1$, and
\begin{align*}
\|\widehat{\bfW}_d-\bfB^{(u)}_{d}O_2O_1^{\top}\|_2^2&\leq \|\widehat{\bfW}_d-\bfB^{(u)}_{d}O_2O_1^{\top}\|_F^2\\
&=\mathrm{tr}\{(\widehat{\bfW}_d-\bfB^{(u)}_{d}O_2O_1^{\top})^{\top}(\widehat{\bfW}_d-\bfB^{(u)}_{d}O_2O_1^{\top})\}\\
&=2\cdot\mathrm{tr}\{\mathbf{I}_d-O_1^{\top}\widehat{\bfW}_d^{\top}\bfB^{(u)}_{d}O_2\}\\
&\leq2\|\sin\Theta\big(\widehat{\bfW}_d,\bfB^{(u)}_{d}\big)\|_F^2.
\end{align*}
Recall that $\widehat{\bfW}_d$ and $\bfB^{(u)}_{d}$ are the top-$d$ eigenvectors of $\widehat{\Sigma}$ and $\bfB\bfB^{\top}$, respectively. The Davis-Kahan theorem \citep{davis1970rotation,yu2015useful} yields that $$\|\sin\Theta\big(\widehat{\bfW}_d,\bfB_{d}^{\LRs{u}}\big)\|_F\leq\dfrac{\|\widehat{\Sigma}-\bfB\bfB^{\top}\|_{F}}{\|\bfbeta_d\|_2^2-\|\bfbeta_{d+1}\|_2^2}.$$ Under (D4'), $\|\bfbeta_d\|_2^2-\|\bfbeta_{d+1}\|_2^2\geq C\cdot p$ for some constant $C>0$. Under (D1)--(D3), standard arguments \citep[see, e.g.,][]{fan2013large,fan2023factor} yields 
\begin{equation*}
\|\widehat{\Sigma}-\bfB\bfB^{\top}\|_F=O_{\P }\LRs{p\delta},
\end{equation*}
where $\delta=n^{-1/2}+(n/\log p)^{-1/2}+p^{-1/2}$. Therefore,
\begin{equation*}  \|\sin\Theta\big(\widehat{\bfW}_d,\bfB_{d}^{\LRs{u}}\big)\|_F=O_{\P }\LRs{\delta},
\end{equation*}
which implies $\|\bfW_d^{\top}\bfB_{-d}\|_2=O_{\P }\LRs{p\delta}$. Combining the preceding bounds gives
\[
\|p^{-1}\bfH_d^{-1}\bfW_d^{\top}\bfB_{-d}\bff_{-\LRm{d}}\|_{L_2(\P)}=O_{\P }\LRs{\delta}.
\]

We next bound the second term on the right-hand side of \eqref{eq:first_term}. From the proof of Lemma 2 in \cite{fan2023factor}, under (D1)--(D3), we have $\|\bfW_d^{\top}\bfu\|^2_{L_2(\P)}=O\LRs{p}$. Together with $\|\bfH_d^{-1}\|_2=O_{\P }(1)$, this yields $\|p^{-1}\bfH_d^{-1}\bfW_d^{\top}\bfu\|_{L_2(\P)}=O_{\P }(p^{-1/2})$. 

Substituting these bounds establishes that, with $O=\bfH_d^{-1}$, $\|O\widehat{\bfR}_{\LRm{d}}-\bff_{\LRm{d}}\|_{L_2(\P)}=O_{\P }\LRs{\delta}$, as claimed.
\end{proof}

In particular, Proposition \ref{prop:dra_factor} implies that the requirement $\|O \widehat{\bfR}_{[d]}-\bff_{[d]}\|_{L_2(\P)} =o_{\P}(n^{-1/4})$ holds with $O=\bfH_d^{-1}$ provided that $(\log p)/\sqrt{n}=o(1)$ and $\sqrt{n}/p=o(1)$.

\subsubsection{Sufficient dimension reduction}

We show that, under suitable conditions, the representation error $\|O\widehat{\bfR}_{[d]}-\bfB_d^{\top}\bfX\|=o_{\P}(n^{-1/4})$ for some data-dependent rotation matrix $O$, when $d<d^*$. We focus on kernel-matrix-based SDR methods (as in the ``inverse regression" class described in Section \ref{sec:ex2null}). Let $M\in\mathbb{R}^{p\times p}$ be the population kernel matrix, whose $d^*$ nonzero eigenvalues determines the SDR subspace. Let $\bfB\in\mathbb{R}^{p\times d^*}$ collect the corresponding eigenvectors, ordered by decreasing eigenvalues. For $1\le d\le d^*$, let $\bfB_d$ denote the first $d$ columns of $\bfB$. Applying the same SDR method to $\{(\bfX_i,Y_i):i\in I_{-k}\}$ yields an empirical kernel matrix $\widehat{M}$, whose leading $d_{\max}$ eigenvectors form $\widehat{\bfB}_{d_{\max}}$. For any $d\le d_{\max}$, write $\widehat{\bfB}_d$ for the first $d$ columns of $\widehat{\bfB}_{d_{\max}}$. The following result parallels Proposition \ref{prop:null_SDR} in the regime $d<d^*$.

\begin{prop}\label{prop:alter_SDR}
Let $1\le d< d^* $. There exists $O=\bfB_d^{\top}\widehat{\bfB}_d\in\mathbb{R}^{d\times d} $ such that \begin{equation*}
\|O\widehat{\bfR}_{\LRm{d}}-\bfB_d^{\top}\bfX\|\leq \|P_{\bfB_d}-P_{\widehat{\bfB}_d}\|_2 \|\bfX\|.
\end{equation*}
\end{prop}

\begin{proof}
See the proof of Proposition \ref{prop:null_SDR}. 
\end{proof}

Using the identity $\|P_{\widehat{\bfB}_d}-P_{\bfB_d}\|_F=2^{1/2}\|\sin\Theta(\widehat{\bfB}_d,\bfB_d)\|_F$, the Davis-Kahan theorem \citep{davis1970rotation,yu2015useful} implies $$\|\sin\Theta(\widehat{\bfB}_d,\bfB_d)\|_F\leq\frac{\|\widehat{M}-M\|_F}{\lambda_d-\lambda_{d+1}},$$ where $\lambda_1,\dots,\lambda_{d^*}$ are the nonzero eigenvalues of $M$, ordered decreasingly. When $p<n$ is fixed, prior work \citep{hsing1992asymptotic,zhu1995asymptotics,li2007directional} established that $\|\widehat{M}-M\|_F=O_{\P}(n^{-1/2})$. Assume further that eigen-gap condition $\lambda_j-\lambda_{j+1}>c>0$ holds for all $1\le j<d^*$, as in \cite{zhu1995asymptotics,yu2025testing}. Then for any $d<d^*$,  $\|\sin\Theta(\widehat{\bfB}_d,\bfB_d)\|_F=O_{\P }(n^{-1/2})$, $\|P_{\bfB_d}-P_{\widehat{\bfB}_d}\|_2=O_{\P}(n^{-1/2})$. Combining this with Proposition \ref{prop:alter_SDR} and assuming $\|\bfX\|$ is bounded, we conclude that $\|O\widehat{\bfR}_{[d]}-\bfB_d^{\top}\bfX\|=o_{\P}(n^{-1/4})$ as desired.

\section{Additional experiments and implementation details}

\subsection{Task-specific construction of the reduction map}\label{suppsec:reduction map}

We use Examples \ref{eg:factor regression}--\ref{eg:reduced rank regression} to illustrate how to construct the reduction map $\widehat{\varphi}$ under different dimension-reduction paradigms and prediction tasks.

\textbf{Factor regression.} A large literature provides procedures for constructing $\widehat{\varphi}$, including principal component analysis \citep{jolliffe2002principal}, maximum likelihood approaches \citep{bai2012factor,doz2012quasi}, and covariance-based methods \citep{fan2013large}. Projected-PCA \citep{fan2016projected} incorporates covariate information to improve factor estimation. \cite{fan2022learning} further develops constructions tailored to specific structures, including characteristic-based weights, moving-window estimators, initial transformations, and Hadamard projections. The choice of method should be guided by the data-generating setting; for example, when informative covariates are available, characteristic-based weighting schemes or projected-PCA may be preferable.

\textbf{Sufficient dimension reduction.} The dimension-reduction method should be chosen to match the prediction task (equivalently, the loss function). As discussed in Section \ref{subsec:d*_d0}, different losses correspond to different target SDR subspaces, and estimation procedures should be selected accordingly. For instance, squared loss targets the central mean subspace, for which principal Hessian direction \citep{li1992principal}, the minimum average variance estimator \citep{xia2002adaptive}, iterative Hessian transformation \citep{cook2002dimension} are appropriate. Negative log-likelihood and cross-entropy losses target the central subspace, for which a broad range of inverse- and forward-regression methods is available \citep{li1991sliced,cook1991sliced,li2007directional,xia2007constructive,wang2008sliced,yin2011sufficient}. Under 0-1 loss, the target is the central discriminant subspace, which can be estimated using methods such as \cite{cook2001theory}.

\textbf{Reduced-rank regression.} Many estimators are available for fitting a low-rank coefficient matrix $\widehat{A}\in\mathbb{R}^{p\times q}$ under a rank constraint \citep{reinsel2022multivariate,yuan2007dimension,bunea2011optimal,zou2022estimation}. Any such method can be used to obtain a rank-$d_{\max}$ estimate $\widehat{A}_{d_{\max}}$. Applying a singular value decomposition to $\widehat{A}_{d_{\max}}$ yields its top $d_{\max}$ left singular vectors, denoted by $\widehat{\bfB}_{d_{\max}}\in\mathbb{R}^{p\times d_{\max}}$. In this framework, one may take $\widehat{\varphi}=\widehat{\bfB}_{d_{\max}}$.

\subsection{Implementation details for the information-criterion baseline}

We compare our method with the information criterion (IC), specifically $IC_{p1}$, proposed by \cite{bai2002determining}. Let $\mathbf{X}=[\bfX_1,\dots,\bfX_n]\in\mathbb{R}^{p\times n}$ denote the data matrix, and let $\mathbf{F}^k=[\bff_1,\dots,\bff_n]\in\mathbb{R}^{k\times n} $ be a matrix of $k$ latent factors. Define $$V(k)=\min_{\bfB^k\in\mathbb{R}^{p\times k},\mathbf{F}^k}\frac{\|\mathbf{X}-B^k\mathbf{F}^k\|_F^2}{np}.$$ The criterion is
\begin{equation}\label{eq:IC}
IC_{p1}(k)=\ln{(V(k))}-k\frac{n+p}{np}\ln{\Big(\frac{n+p}{np}\Big)}.
\end{equation}
Equivalently, $V(k)$ can be computed as the sum of the eigenvalues of $\mathbf{X}\mathbf{X}^{\top}/(np)$ from $(k+1)$ to $p$.

\subsection{Sensitivity to the choice of learner class}\label{sec:ftm simulation}

\begin{figure}[!h]
\centering
\includegraphics[width=0.7\linewidth]{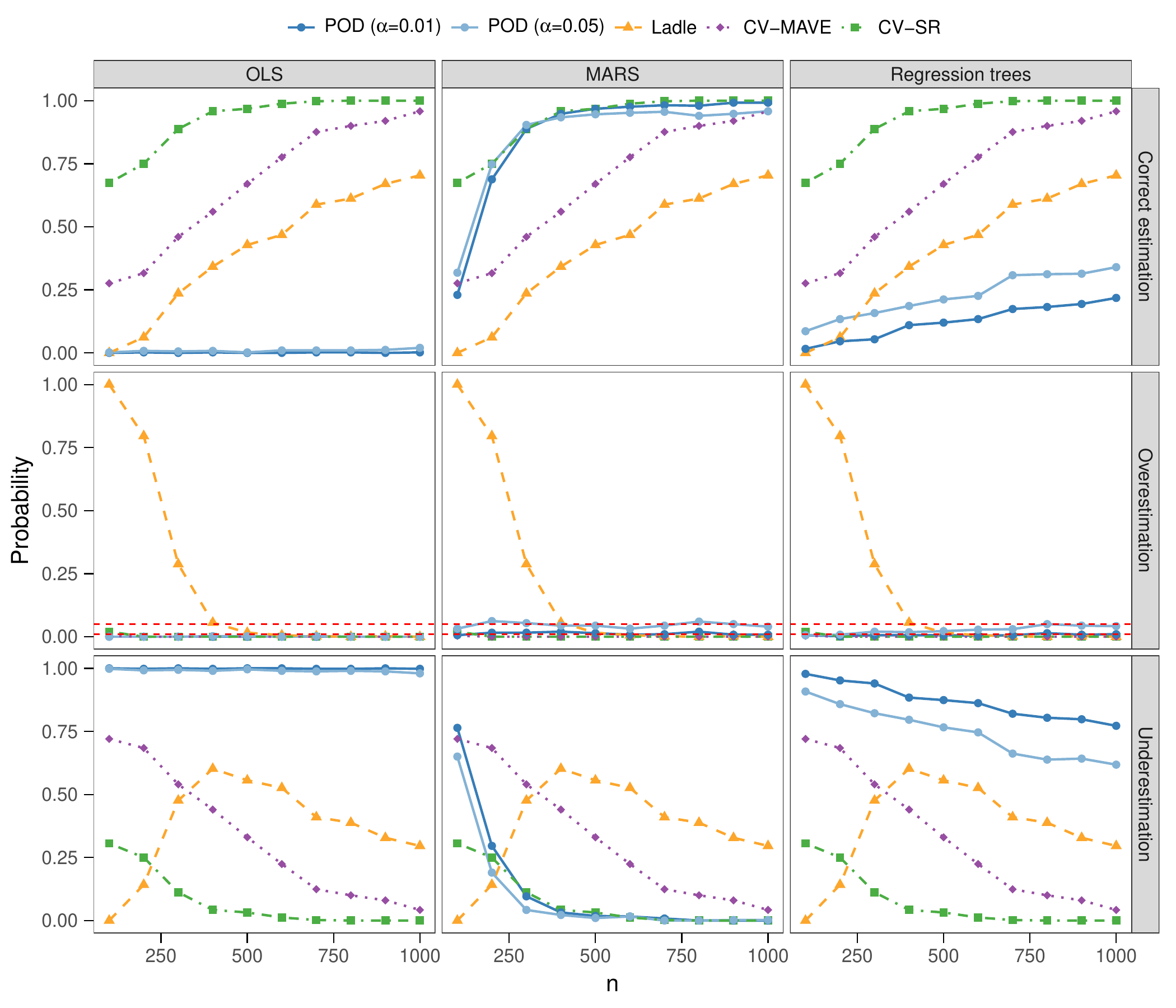}
\caption{Empirical probabilities of correct estimation ($\widehat{d}=d^*_{\rm CMS}$), overestimation ($\widehat{d}>d^*_{\rm CMS}$), and underestimation ($\widehat{d}<d^*_{\rm CMS}$) across $n$ for SDR. Red dashed horizontal lines in the second row mark the nominal significance levels $\alpha=0.01$ and $\alpha=0.05$.} 
\label{Fig: percentages of correct order estimation with df}
\end{figure}

We use Model 4 in Section \ref{subsec: simulation for SDR} to examine how the choice of $\mathcal{F}$ affects POD. Keeping all other specifications fixed as in the main text, we restrict $\mathcal{F}$ to one of the following learner classes: OLS, MARS, or regression trees. The resulting performance of each method is reported in Figure~\ref{Fig: percentages of correct order estimation with df}.

OLS and regression trees do not adequately approximate the relationship between $\bfY$ and $(\bfX_1$, $\bfX_2)$ in this model, reflecting limitations of these learner classes in capturing the underlying functional form; this misspecification in turn degrades POD's performance. In contrast, MARS provides a more accurate approximation and yields satisfactory results. These findings highlight the importance of choosing a sufficiently rich learner class $\mathcal{F}$, one that includes flexible models capable of representing the underlying structure.

\subsection{SDR with categorical responses}\label{sec:categorical response}

For categorical responses, let $X\sim \mathcal{N}(0,I_p)$ with $p=10$ and consider
\begin{align}
&\text{Model 6}:\quad Y\sim\text{Binomial}(2,\text{logit}(X_1)), \nonumber\\
&\text{Model 7}:\quad Y=\mathbf{1}\LRl{X_1+\dots+X_5+\sigma\varepsilon>1}+2\cdot\mathbf{1}\LRl{X_6+\dots+X_{10}+\sigma\varepsilon>0}, \nonumber
\end{align}
where $\varepsilon\sim\mathcal{N}(0,1)$ is independent of $X$ and $\sigma=0.5$. The structural dimensions are $d^*_{\rm CS}=1$ (Model 6) and $d^*_{\rm CS}=2$ (Model 7, \cite{zhu2006fourier}); in both models, $d^*_{\rm CS}=d^*_{\rm CMS}$. POD evaluates predictiveness using cross-entropy loss, so that $d^*(\V)=d^*_{\rm CS}$. For POD, we estimate the reduction maps for Models~6--7 using SIR, with $3$ slices for Model~6 and $4$ slices for Model~7. The downstream classifier is selected by a two-fold cross-validation between a support vector machine and a classification tree.

We compare POD with three alternatives: the ladle estimator based on SIR, and cross-validation criteria based on MAVE and SR. All methods target $d^*_{\rm CS}=1$ for Model~6 and $d^*_{\rm CS}=2$ for Model~7.

\begin{figure}[!h]
\centering
\includegraphics[width=0.6\linewidth]{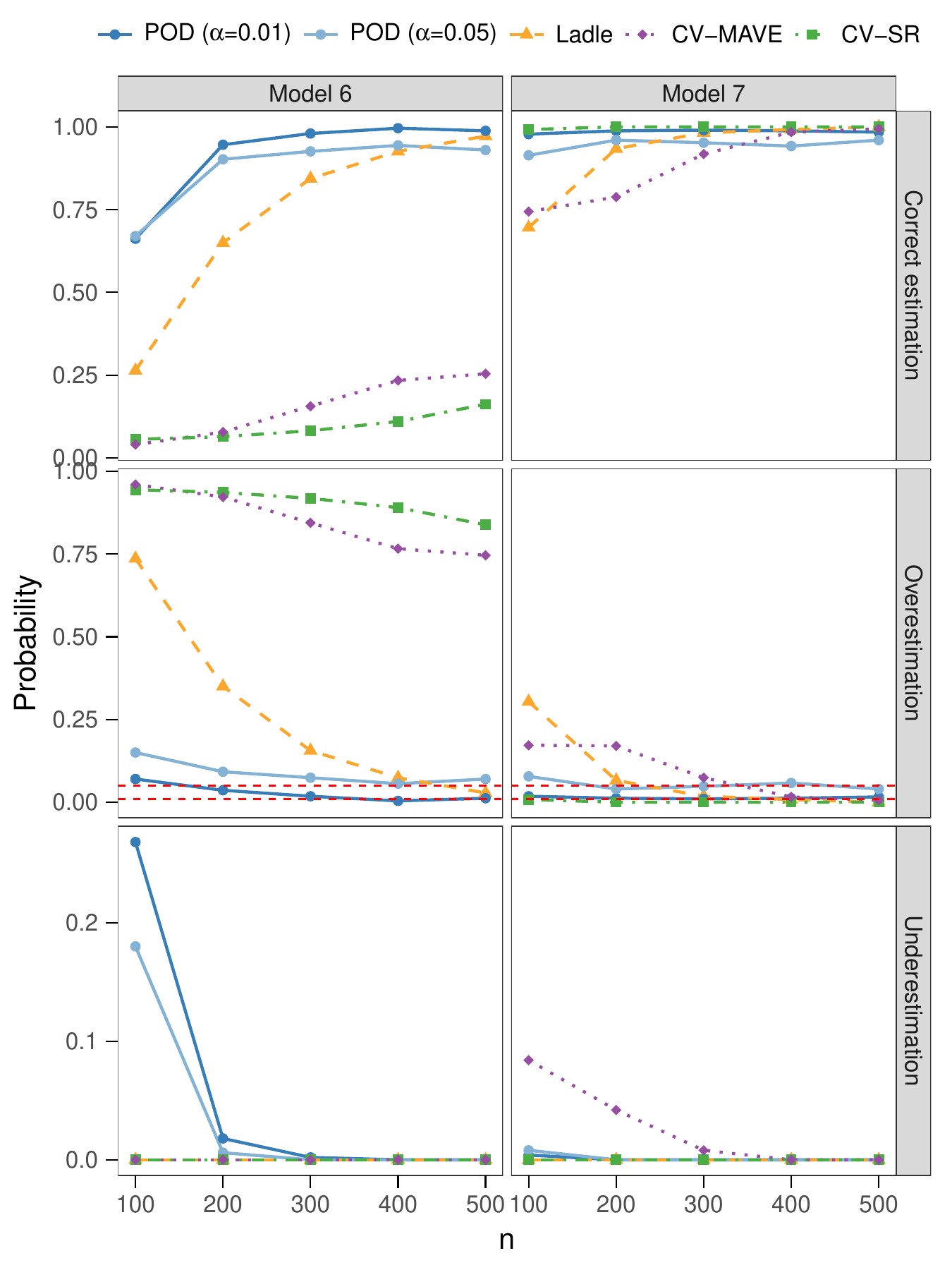}
\caption{Empirical probabilities of correct estimation ($\widehat{d}=d^*_{\rm CS}$), overestimation ($\widehat{d}>d^*_{\rm CS}$), and underestimation ($\widehat{d}<d^*_{\rm CS}$) across $n$ for SDR Models 6--7. Red dashed horizontal lines in the second row mark the nominal significance levels $\alpha=0.01$ and $\alpha=0.05$.} 
\label{Fig: percentages of CS estimation}
\end{figure}

Figure~\ref{Fig: percentages of CS estimation} reports the empirical probabilities of correct estimation, overestimation, and underestimation of $d^*_{\rm CS}$ as $n$ increases. Across both models, POD converges more rapidly than the competitors. Its overestimation probability approaches the nominal level, while the underestimation probability vanishes.

\subsection{Selection-frequencies in the real-data analysis}\label{sec:frequency}

Because our method uses $K$-fold cross-fitting, the ladle estimator relies on bootstrapping, and the critical value of the weighted-$\chi^2$ test is obtained via simulation, all methods exhibit randomness even when applied to the same dataset. We therefore report, in Tables~\ref{tab:frequency_dmax8}--\ref{tab:frequency_dmax16}, the selection frequencies of each estimated dimension across 100 repeated runs, which demonstrate that POD yields stable dimension selection.

\begin{table}[!h]
\renewcommand{\arraystretch}{1.25}
\centering
\caption{Selection frequencies of estimated dimensions $\hat{d}$ over 100 repetitions when $d_{\max}=8$.}
\label{tab:frequency_dmax8}
{ \footnotesize
\begin{tabular}{lrrrcrrcrrrcrrcrcrrcr}
\toprule
& \multicolumn{13}{c}{POD}                                                                    &                                         & Ladle & & \multicolumn{4}{c}{BY, weighted-$\chi^2$}                                  \\
& \multicolumn{6}{c}{0-1 loss}                                    &  & \multicolumn{6}{c}{Cross-entropy}                             & &  &  &     &                &               &                           \\ \cline{2-7} \cline{9-14}  \cline{16-16} \cline{18-21}
& \multicolumn{3}{c}{$\alpha=1\%$} & &\multicolumn{2}{c}{$\alpha=5\%$} &  & \multicolumn{3}{c}{$\alpha=1\%$} & & \multicolumn{2}{c}{$\alpha=5\%$} &   & &    & \multicolumn{2}{c}{$\alpha=1\%$} & &  \multicolumn{1}{c}{$\alpha=5\%$}                \\  
$\widehat{d}$         & $1$   & $2$    & $ 3$ & & $2$       & $\geq 3$     &  & $1$   & $2 $   & $ 3$&  & $2$       &   $\geq 3$ &  & $3$ &    & $14$             & $15 $   &        & $16$    \\
Frequency & $1 $  &$ 93 $  & $6$ &                & $89$      & $11 $                  &  & $1 $  & $97$   & $2$      &           & $94$      & $6$                  &  & $100$&   & $81$             & $19$ &           & $100$   \\ \bottomrule
\end{tabular}}
\end{table}

\begin{table}[!h]
\renewcommand{\arraystretch}{1.25}
\centering
\caption{Selection frequencies of estimated dimensions $\hat{d}$ over 100 repetitions when $d_{\max}=16$.}
\label{tab:frequency_dmax16}
{\footnotesize
\begin{tabular}{lrrcrrcrrcrrcrrcrrcr}
\hline
& \multicolumn{11}{c}{POD}  &  & \multicolumn{2}{c}{Ladle} &  & \multicolumn{4}{c}{BY, weighted-$\chi^2$}                                        \\
& \multicolumn{5}{c}{0-1 loss}                                          &  & \multicolumn{5}{c}{Cross-entropy}                                     &  & \multicolumn{2}{c}{}      &  & \multicolumn{4}{c}{}                                               \\ \cline{2-6} \cline{8-12} \cline{14-15} \cline{17-20} 
& \multicolumn{2}{c}{$\alpha=1\%$}    &  & \multicolumn{2}{c}{$\alpha=5\%$} &  & \multicolumn{2}{c}{$\alpha=1\%$}    &  & \multicolumn{2}{c}{$\alpha=5\%$} &  & \multicolumn{2}{c}{}      &  & \multicolumn{2}{c}{$\alpha=1\%$} &  & \multicolumn{1}{l}{$\alpha=5\%$} \\
$\widehat{d}$         & $2$  &$3$        &  & $2$       & $\geq 3$     &  & $2$  & $3$        &  & $2$       & $\geq 3$     &  & $\leq 13$    & $14$    &  & $14$             & $15$            &  & $16$                             \\
Frequency & $93$ & $7$        &  & $84$      & $16$                   &  & $95$& $5$       &  & $94$     & $6$                    &  & $13$                & $87$    &  & $81$             & $19$            &  & $100$                            \\ \hline
\end{tabular}}
\end{table}

\end{document}